\documentclass[11pt]{article}
\usepackage[utf8]{inputenc}
\usepackage[margin=1in]{geometry}
\usepackage[shortlabels]{enumitem}
\usepackage{bm}
\usepackage{amsfonts}
\usepackage{amsmath}
\usepackage{xcolor}
\usepackage{graphicx}
\usepackage{amsthm}
\usepackage{qtree}
\usepackage{tree-dvips}
\usepackage{float}
\usepackage{authblk}
\usepackage{hyperref}
\newtheorem{theorem}{Theorem}
\newtheorem*{theorem*}{Theorem}
\newtheorem{lemma}{Lemma}
\newtheorem*{lemma*}{Lemma}
\newtheorem{corollary}{Corollary}

\newtheorem{observation}{Observation}
\newtheorem{definition}{Definition}

\usepackage[numbers,sort&compress]{natbib}

\newcommand{\varbound}{\frac{2^n}{2n}}
\newcommand{\ttbound}{\frac{N}{2\log N}}
\def\cl{\bf}
\author{Oliver Korten \thanks{Department of Computer Science, Columbia University. Email: \href{mailto:oliver.korten@columbia.edu}{oliver.korten@columbia.edu}}}
\date{May 2021}
\title{The Hardest Explicit Construction}
\begin{document}
\maketitle
 
\begin{abstract}
We investigate the complexity of explicit construction problems, where the goal is to produce a particular object of size $n$ possessing some pseudorandom property in time polynomial in $n$. We give overwhelming evidence that {\cl APEPP}, defined originally by Kleinberg et al. \cite{total}, is the natural complexity class associated with explicit constructions of objects whose existence follows from the probabilistic method, by placing a variety of such construction problems in this class.
We then demonstrate that a result of Je\v{r}\'{a}bek \cite{Jerabek} on provability in Bounded Arithmetic, when reinterpreted as a reduction between search problems, shows that constructing a truth table of high circuit complexity is complete for {\cl APEPP} under ${\cl P}^{\cl NP}$ reductions. This illustrates that Shannon's classical proof of the existence of hard boolean functions is in fact a \emph{universal} probabilistic existence argument: derandomizing his proof implies a generic derandomization of the probabilistic method.
As a corollary, we prove that  ${\cl EXP}^{\cl NP}$ contains a language of circuit complexity $2^{n^{\Omega(1)}}$ if and only if it contains a language of circuit complexity $\varbound$.
Finally, for several of the problems shown to lie in {\cl APEPP}, we demonstrate direct polynomial time reductions to the explicit construction of hard truth tables.
\end{abstract}
\newpage
\section{Introduction}\label{sec:intro}
Explicit construction --- the task of replacing a nonconstructive argument for the existence of a certain type of object with a deterministic algorithm that outputs one --- is an important genre of computational problems, whose history is intertwined with the most fundamental questions in complexity and derandomization. The primary method of existence argument for these problems is to show that a random object has a desired property with high probability. This technique, initiated by Erd{\"o}s \cite{Erdos} and since dubbed the ``probabilistic method,'' has proven immensely useful across disparate subfields of combinatorics and computer science. Indeed, the probabilistic method is currently our sole source of certainty that there exist hard Boolean functions, pseudorandom number generators, rigid matrices, and optimal randomness extractors, 
among a variety of other combinatorial objects. 

Explicit construction problems can be phrased, in complexity terms, as sparse search problems: given the input $1^n$, output some object of size $n$ satisfying a certain property. In the interesting case, such problems are also \emph{total}: we have a reason to believe that for all $n$, at least one object with this property exists. In contrast to the fundamental importance of explicit constructions, there has been surprisingly little work attempting to systematically study their complexity. This gap was pointed out previously by Santhanam \cite{explicit-constructions}, who studied the complexity of explicit construction problems from the following perspective: we have some property $\Pi$ which is \emph{promised} to hold for almost all strings of length $n$. Based on the complexity of testing the property $\Pi$, what can be said about the complexity of producing an $n$-bit string with property $\Pi$? Though some interesting reductions can be shown in this framework, Santhanam notes that this approach does not seem to yield robust complexity classes with \emph{complete} explicit construction problems.

This issue is familiar in the study of the class {\cl TFNP}: when we have only a \emph{promise} that a search problem is total, it is seemingly impossible to reduce it to a problem of similar complexity which has a \emph{syntactic} guarantee of totality. This led to the study, initiated by Papadimitriou \cite{christos}, of characterizing total search problems based on the combinatorial lemma which guarantees the existence of a solution. In recent work of Kleinberg et al. \cite{total}, this method was used to analyse several total search problems in the polynomial hierarchy beyond {\cl NP}. One class they define is {\cl APEPP}, which consists of the ${\cl \Sigma_2^P}$ total search problems whose totality follows from the ``Abundant Empty Pigeonhole Principle,'' which tells us that a function $f: \{0,1\}^n \rightarrow \{0,1\}^{n+1}$ cannot be surjective. In this paper, we show that {\cl APEPP} is the natural \emph{syntactic} class into which we can place a vast range of explicit construction problems where a solution is guaranteed by the probabilistic method.

Given that {\cl APEPP} is a syntactic class, it is natural to ask whether some explicit construction problem is complete for it. As it turns out, the answer is positive: constructing a truth table of length $2^n$ with circuit complexity $2^{\epsilon n}$ is in fact \emph{complete} for {\cl APEPP} under ${\cl P}^{\cl NP}$ reductions. Perhaps surprisingly, this important fact had been known for many years in the universe of Bounded Arithmetic, essentially proved in Emil Je\v{r}\'{a}bek's PhD thesis in 2004. Here Je\v{r}\'{a}bek shows that the theorem asserting the empty pigeonhole principle is \emph{equivalent}, in a particular theory of Bounded Arithmetic, to the theorem asserting the existence of hard boolean functions. Although his result is phrased in terms of logical expressibility, we show that when translated to language of search problems his techniques give a ${\cl P}^{\cl NP}$ reduction from any problem in {\cl APEPP} to the problem of constructing a hard truth table. In Section~\ref{sec:completeness} we give a self-contained proof of this, and generalize the reduction to hold for arbitrary classes of circuits equipped with oracle gates. Combined with our results placing a wide range of explicit construction problems in {\cl APEPP}, this shows that in a concrete sense, constructing a hard truth table is a \emph{universal explicit construction problem}. We give further credence to this claim by showing in addition that several well known explicit construction problems in {\cl APEPP,} including the explicit construction of rigid matrices, can be directly reduced to the problem of constructing a hard truth table via polynomial time reductions (as opposed to ${\cl P}^{\cl NP}$ reductions).

\subsection{Our Contributions}
We investigate the complexity class {\cl APEPP} introduced in \cite{total}, defined by the following complete problem {\sc Empty}: given a circuit $C: \{0,1\}^n \rightarrow \{0,1\}^m$ with $m > n$, find an $m$-bit string outside the range of $C$. In Section~\ref{sec:constructions-in-apepp} we give overwhelming evidence that {\cl APEPP} is the natural class associated with explicit constructions from the probabilistic method, by placing a wide range of well studied problems in this class. In particular, we show that the explicit construction problems associated with the following objects lie in {\cl APEPP}:
\begin{itemize}
    \item Truth tables of length $2^n$ with circuit complexity $\varbound$ (Theorem~\ref{thm:truth-tables-in-apepp})
    \item Pseudorandom generators (Theorem~\ref{thm:prg})
    \item Strongly explicit two-source randomness extractors with 1 bit output for min-entropy $\log n + O(\log(1/\epsilon))$, and thus strongly explicit $O(\log n)$-Ramsey graphs in both the bipartite and non-bipartite case (Theorem~\ref{thm:ext})
    \item Matrices with high rigidity over any finite field (Theorem~\ref{thm:rigid-in-apepp})
    \item Strings of time-bounded Kolmogorov complexity $n - 1$ relative to any fixed polynomial time bound and any fixed Turing machine (Theorem~\ref{thm:kt-random})
    \item Communication problems outside of ${\cl PSPACE^{CC}}$ (Theorem~\ref{thm:space-circuit})
    \item Hard data structure problems in the-bit probe model (Theorem~\ref{thm:probe-circuit})
\end{itemize}

Since the work of Impagliazzo and Wigderson \cite{IW} implies that constructing pseudorandom generators reduces to constructing hard truth tables, {\cl APEPP} constructions of PRGs follow immediately from {\cl APEPP} constructions of hard truth tables. However, we provide a self-contained and simple proof that PRG construction can be reduced to {\sc Empty}, without requiring the more involved techniques of Nisan, Wigderson, and Impagliazzo \cite{NW}\cite{IW}. Together with the result in the following section that constructing hard truth tables is complete for {\cl APEPP} under ${\cl P}^{\cl NP}$ reductions, this gives an alternative and significantly simplified proof that worst-case-hard truth tables can be used to derandomize algorithms (although it proves a weaker result, that this derandomization can be accomplished with an {\cl NP} oracle).

In Section~\ref{sec:completeness} we show that constructing a truth table of length $2^n$ with circuit complexity $2^{\epsilon n}$ is complete for {\cl APEPP} under ${\cl P}^{\cl NP}$ reductions (for any fixed $0 < \epsilon < 1)$. As discussed earlier, the core argument behind this result was proven by Je\v{r}\'{a}bek in \cite{Jerabek}, where he shows that the theorem asserting the existence of hard boolean functions is equivalent to the theorem asserting the empty pigeonhole principle in a certain fragment of Bounded Arithmetic. We show that, when viewed through the lens of explicit construction problems, this technique yields a reduction from {\sc Empty} to the explicit construction of hard truth tables. We also generalize this reduction to arbitrary oracle circuits, which allows us to prove the following more general statement: constructing a truth table which requires large $\Sigma^{\cl P}_i$-oracle circuits is complete for {\cl APEPP}$_{\Sigma^{\cl P}_i}$ under $\Delta^{\cl P}_{i+2}$ reductions (the complete problem for {\cl APEPP}$_{\Sigma^{\cl P}_i}$ is the variant of {\sc Empty} where the input circuit can have $\Sigma^{\cl P}_i$-oracle gates). By recasting and generalizing Je\v{r}\'{a}bek's theorem in the context of explicit construction problems, we are able to derive several novel results. First and foremost, we conclude that there is a ${\cl P}^{\cl NP}$ construction of hard truth tables \emph{if and only if} there is a ${\cl P}^{\cl NP}$ algorithm for every problem in {\cl APEPP}, and so in particular such a construction of hard truth tables would automatically imply ${\cl P}^{\cl NP}$ constructions for each of the well-studied problems discussed in Section~\ref{sec:constructions-in-apepp}. This tells us that constructing hard truth tables is, in a definite sense, \emph{a universal explicit construction problem}. Since the existence of a ${\cl P}^{\cl NP}$ construction of hard truth tables is equivalent to the existence of a language in ${\cl E}^{\cl NP}$ with circuit complexity $2^{\Omega(n)}$, this completeness result actually gives an \emph{exact algorithmic characterization} of proving $2^{\Omega(n)}$ circuit lower bounds for ${\cl E}^{\cl NP}$:

\begin{theorem*}[Theorem~\ref{thm:LB-algo-equivalence}]
There is a ${\cl P}^{\cl NP}$ algorithm for {\sc Empty} if and only if ${\cl E}^{\cl NP}$ contains a language of circuit complexity $2^{\Omega(n)}$.
\end{theorem*}
As a corollary we are able to derive the following:
\begin{theorem*}[Corollaries \ref{cor:e-hard} and \ref{cor:exp-hard}]
${\cl E}^{\cl NP}$ (resp. ${\cl EXP}^{\cl NP}$)  contains a language of circuit complexity $2^{\Omega(n)}$ (resp. $2^{n^{\Omega(1)}}$) \emph{if and only if} ${\cl E}^{\cl NP}$ (resp. ${\cl EXP}^{\cl NP}$) contains a language of circuit complexity $\varbound$.
\end{theorem*}

Unpacking the proof of the above corollaries reveals an efficient algorithm to ``extract hardness'' from truth tables using an oracle for circuit minimization, a prospect previously considered in \cite{hardness-harder}:

\begin{theorem*}[Theorem \ref{thm:hardness-extractor}]
There is a polynomial time algorithm using a circuit minimization oracle (or more generally an {\cl NP} oracle) which, given a truth table $x$ of length $M$ and circuit complexity $s$, outputs a truth table $y$ of length $N = {\Omega}(\sqrt{\frac{s}{\log M}})$ and circuit complexity $\Omega(\frac{N}{\log N})$.
\end{theorem*}

We then argue, based on an observation of Williams \cite{sat-lb}, that improving this construction in its current form to extract $(\frac{s}{\log M})^{\frac{1}{2} + \epsilon}$ bits of hardness would require a breakthrough for {\sc 3SUM}.

Finally, in Section~\ref{sec:direct-reductions} we consider {\cl P} (as opposed to ${\cl P}^{\cl NP}$) reductions from particular explicit construction problems to the problem of constructing hard truth tables. We show that in the case of rigidity, bit probe lower bounds, and certain communication complexity lower bounds, such reductions exist. These reductions take the following form: we show that the failure of an $n$-bit string $x$ to satisfy certain pseudorandom properties implies a smaller than worst case circuit computing $x$. This then implies that any $n$-bit string of sufficiently high circuit complexity will necessarily possess a variety of pseudorandom properties, including high rigidity, high space-bounded communication complexity, and high bit-probe complexity. We also make note of an interesting dichotomy (Theorem~\ref{thm:dichotomy}), which tells us that any explicit construction problem in {\cl APEPP} is either {\cl APEPP}-complete under ${\cl P}^{\cl NP}$ reductions, or solvable in subexponential time with an {\cl NP} oracle (for infinitely many input lengths).

Another concrete takeaway from this work is that we demonstrate, for several well-studied problems, the weakest known assumptions necessary to obtain explicit constructions of a certain type (polynomial time constructions in some cases and ${\cl P}^{\cl NP}$ constructions in others). Perhaps the most interesting application of this is rigidity, as the complexity of rigid matrix construction has been studied extensively in both the {\cl P} and ${\cl P}^{\cl NP}$ regimes. We obtain the following conditional constructions of rigid matrices:

\begin{theorem*}[Theorems \ref{thm:rigid-in-apepp} and \ref{thm:LB-algo-equivalence}]
If ${\cl E}^{\cl NP}$ contains a language of circuit complexity $2^{\Omega(n)}$, then for any prime power $q$ there is a ${\cl P}^{\cl NP}$ construction of an $n \times n$ matrix over $\mathbb{F}_q$ which is $\Omega(n^2)$-far (in hamming distance) from any rank-$\Omega(n)$ matrix.
\end{theorem*}

\begin{theorem*}[Theorem \ref{thm:rigidity-circuit}]
If ${\cl E}$ contains a language of circuit complexity $\Omega(\frac{2^n}{n})$, then there is a polynomial time construction of an $n \times n$ matrix over $\mathbb{F}_2$ which is $\Omega(n^2)$-far from any rank-$\Omega(n)$ matrix.
\end{theorem*}

In both cases, the rigidity parameters in the conclusion would be sufficient to carry out Valiant's lower bound program \cite{Valiant}. The weakest hardness assumptions previously known to yield constructions with even remotely similar parameters (in either the ${\cl P}^{\cl NP}$ or {\cl P} regimes) require a lower bound against \emph{nondeterministic} circuits \cite{AM-SV}.

\subsection{Related Work}
A large body of work on the hardness/randomness connection, starting with that of Nisan and Wigderson \cite{NW}, has exhibited the usefulness of explicit constructions of hard truth tables. The results of Impagliazzo and Wigderson \cite{IW} give, in particular, a reduction from explicit constructions of hard truth tables to explicit constructions of pseudorandom generators that fool polynomial size circuits. As noted by Santhanam \cite{explicit-constructions}, this immediately implies that for any ``dense'' property $\Pi$ recognizable in {\cl P} (dense meaning the fraction of $n$-bit strings holding this property is at least $1/poly(n)$), an efficient construction of a hard truth table immediately implies an efficient construction of an $n$-bit string with property $\Pi$. But many properties of interest such as Rigidity (or any of the other properties studied in this work) are only known to be recognizable in the larger class {\cl NP}. Under the stronger assumption that we can construct truth tables hard for certain classes of \emph{nondeterministic} circuits, constructions for all dense {\cl NP} properties are known to follow as well \cite{AM-derandomization} \cite{AM-SV}, so in particular ${\cl P}^{\cl NP}$ constructions for every problem in {\cl APEPP} would follow. However, constructing truth tables that are hard for nondeterministic circuits appears strictly harder than constructing truth tables hard for standard circuits, and in particular does not seem to be contained in {\cl APEPP}, so although this yields an explicit construction problem which is \emph{hard} for {\cl APEPP}, it does not appear to be complete. In contrast, we show here that constructing a truth table which is hard for standard circuits is both contained in and hard for {\cl APEPP}, thus showing that a ${\cl P}^{\cl NP}$ construction of a hard truth table is possible \emph{if and only if} such a construction is possible for every problem in {\cl APEPP}.

For several of the problems we study, a long line of work has gone into improving state-of-the-art explicit constructions. We give a brief overview here of some recent work on rigid matrices and extractors. Rigidity was first introduced by Valiant \cite{Valiant}, who showed that any matrix which is $n^{1+\epsilon}$-far from a rank-$\delta n$ matrix for some $\epsilon, \delta > 0$ cannot be computed by linear size, logarithmic depth arithmetic circuits. Since then it has been a notorious open problem to provide examples of an explicit matrix family with rigidity parameters anywhere close to this. A recent breakthrough was achieved in \cite{alman-chen}, and improved by \cite{pcp-rigidity}, which gives ${\cl P}^{\cl NP}$ constructions of matrices that are $\Omega(n^2)$-far from any rank-$2^{\log n/\Omega(\log \log n)}$ matrix, for infinitely many values of $n$. However, such a construction is still not sufficient to carry out Valiant's arithmetic circuit lower bound program. A conditional result of \cite{data-structure-rigidity} implies that ${\cl P}^{\cl NP}$ constructions of certain rigid matrices are possible assuming explicit constructions exist for certain  hard data structure problems in the group model. In terms of polynomial time constructions, the best known construction yields a matrix which is $\frac{N^2}{\rho}\log{\frac{N}{\rho}}$-far from any $\rho$-rigid matrix, for any parameter $\rho$ \cite{SSS}\cite{Friedman}.

The case of Ramsey graphs and extractors is slightly more complicated. There are two common definitions of the explicit construction problems corresponding to these objects. The first is often referred to as the ``weakly-explicit'' version, where we must output the adjacency matrix (in the case of Ramsey graphs) or truth table (in the case of extractors) in time polynomial in the size of the truth table/matrix. The second version, referred to as the ``strongly-explicit'' version, is to output a succinct circuit which computes the adjacency relation or extractor function. Clearly the strongly-explicit case is harder, but in both cases, there is a significant gap between what is achievable by explicit methods and what can be proven possible by the probabilistic method. We will focus on the strongly explicit case in this work, and in the case of extractors we will focus on \emph{two-source} extractors with one bit of output, which remains the most challenging current frontier \cite{extractors-survey}. The state of the art constructions for both two-source extractors and Ramsey graphs are due to Li \cite{li-extractor}. He demonstrates a two-source extractor for min-entropy $O(\log n \log \log n)$, and hence an $n$-vertex graph which is $(\log n)^{O(\log \log \log n)}$-Ramsey. Our results show that strongly explicit extractors for min-entropy $\log n + O(\log(1/\epsilon))$, and thus $n$-vertex $O(\log n)$-Ramsey graphs, can be constructed in {\cl APEPP}. In both cases these parameters are known to be the best possible. For a comprehensive survey of recent progress on extractors see \cite{extractors-survey}.

Another line of work in the area of explicit constructions investigates the possibility of \emph{pseudodeterministic} constructions of certain objects. Here, the construction algorithm is allowed to use randomness, but must output the \emph{same object} on most computation paths. Originally introduced in \cite{GG}, this paradigm was recently applied in \cite{pseudo-primes} to the construction of prime numbers, where a subexponential time pseudodeterministic construction which works for infinitely many input lengths is given.

\subsection{Proof Sketch of Main Theorem}
We give here an informal overview of the proof that {\sc Empty} can be solved in polynomial time given access to a hard truth table and an {\cl NP} oracle. At the core of this proof is a familiar construction in the theory of computing which dates back to the 1980's, namely the pseudorandom function generator of Goldreich, Goldwasser, and Micali \cite{GGM}. Note that in the following, we will refer to the ``circuit complexity of an $n$-bit string $x$'' to mean the size of the smallest circuit computing $x_i$ given $i$ in binary; this is well-defined even when $n$ is not a power of 2, as we shall formalize Section~\ref{sec:constructions-in-apepp}\footnote{See Definition~\ref{def:circ-size}}.

Consider the special case of {\sc Empty} where our input is a circuit $C: \{0,1\}^n \rightarrow \{0,1\}^{2n}$ which exactly doubles its input size. For a moment let us forget our primary goal of finding a $2n$-bit string outside $C$'s range, and instead consider $C$ as a cryptographic pseudorandom generator which we are attempting to break. Since $C$ is a function which extends its input size by a positive number of bits, it is indeed of the same syntactic form as a cryptographic PRG, so this viewpoint is well-defined.

In \cite{GGM}, Goldreich, Goldwasser and Micali give a procedure\footnote{\cite{GGM} and \cite{natural-proofs} apply the construction described here in a different parameter regime, so our statement of the result differs slightly from its original presentation. The version described here has been noted subsequently in the literature on {\sc MCSP}, see for example \cite{MCSP-OWF}.} which, for any fixed $0 < \epsilon < 1$, takes $C$ and produces in polynomial time a new circuit $C^*: \{0,1\}^n \rightarrow \{0,1\}^{m}$ for some $m = poly(n)$, which satisfies the following two properties:
\begin{enumerate}[(1)]
    \item Every string in the range of $C^*$ has circuit complexity at most $m^{\epsilon}$.
    \item Given a statistical test breaking $C^*$, we can construct a statistical test of similar complexity breaking $C$.
\end{enumerate}
The construction of $C^*$ is in fact quite simple: for an appropriate choice of $k$, we recursively apply $C$ to an $n$-bit input for $k$ iterations as follows: first apply $C$ to an $n$-bit string to get 2 $n$-bit strings, then apply it again to each of those to get $4$, and continue $k$ times until we obtain $2^k$ $n$-bit strings. 

A key observation made by Razborov and Rudich \cite{natural-proofs} is that condition (1) automatically implies a \emph{particular statistical test which breaks $C^*$}, namely the test which accepts precisely those $m$-bit strings with circuit complexity exceeding $m^{\epsilon}$. But by property (2), $C^*$ inherets the security of $C$, which is an arbitrary candidate PRG. This means that determining if an $m$-bit string has circuits of size $m^{\epsilon}$ is in fact a \emph{universal test for randomness}, capable of simultaneously breaking all pseudorandom generators.

Recall now our original goal for $C$, which was to find a $2n$-bit string outside its range. Property (1) of $C^*$ implies that an explicit construction of a length-$m$ truth table of circuit complexity $m^{\epsilon}$ would immediately yield an explicit $m$-bit string outside the range of $C^*$. In Section~\ref{sec:completeness}, we show that $C^*$ obeys the following third property:
\begin{enumerate}[(1)]
    \setcounter{enumi}{2}
    \item Given a string outside the range of $C^*$, we can find a string outside the range of $C$ using a polynomial number of calls to an {\cl NP} oracle.
\end{enumerate}
The analogue of statement (3) in the context of Bounded Arithmetic was first shown by Je\v{r}\'{a}bek \cite{Jerabek}, and a quite similar argument appears even earlier in the work of Paris, Wilkie, and Woods \cite{Paris-Wilkie-Woods}. Combining properties (1) and (3), we get the desired result: any $m$-bit string of complexity $m^{\epsilon}$ must lie outside the range of $C^*$, so using such a string together with an {\cl NP} oracle we can solve our original instance of {\sc Empty}.

To summarize, the construction $C^*$ of Goldreich, Goldwasser and Micali shows that the property of requiring large circuits is a \emph{universal pseudorandom property of strings} in two concrete senses: 
\begin{enumerate}[(a)]
    \item (Original analysis of \cite{GGM} and \cite{natural-proofs}) A test determining whether a string requires large circuits can be efficiently boostrapped into a test distinguishing any pseudorandom distribution from the uniform distribution.
    \item (This work together with \cite{Jerabek}) An explicit example of a string requiring large circuits can be used to generate an explicit example of a string outside the range of any efficiently computable map $C: \{0,1\}^n \rightarrow \{0,1\}^{2n}$ (in fact any $C: \{0,1\}^n \rightarrow \{0,1\}^{n+1}$ as shown in Section~\ref{sec:completeness}), and in particular can be used to construct explicit examples of strings possessing each of the fundamental pseudorandom properties examined in Section~\ref{sec:constructions-in-apepp}.
\end{enumerate}

\section{Definitions}\label{sec:defs}
 Following \cite{total}, we define the set of total functions in ${\cl \Sigma_2^P}$, denoted ${\cl TF\Sigma_2^P}$, as follows:
\begin{definition}
A relation $R(x,y)$ is in ${\cl TF\Sigma_2^P}$ if there exists a polynomial $p(n)$ such that the following conditions hold:
\begin{enumerate}
    \item For every $x$, there exists a $y$ such that $|y| \leq p(|x|)$ and $R(x,y)$ holds
    \item There is a polynomial time Turing machine $M$ such that \\
    $R(x,y) \Longleftrightarrow \forall z \in \{0,1\}^{p(|x|)} M(x,y,z) \text{ accepts}$ 
\end{enumerate}
\end{definition}
The search problem associated with such a relation is: ``given $x$, find some $y$ such that $R(x,y)$ holds.'' For the majority of this paper, we will be concerned primarily with \emph{sparse} ${\cl TF\Sigma_2^P}$ search problems, where the only relevant part of the input is its length. We can thus define the following ``sparse'' subclass of ${\cl TF\Sigma_2^P}$: 
\begin{definition}
A relation $R(x,y)$ is in ${\cl STF\Sigma_2^P}$ if $R \in {\cl TF\Sigma_2^P}$ and for any $x_1,x_2$ such that $|x_1|=|x_2|$, we have that for all $y$, $R(x_1,y) \Leftrightarrow R(x_2,y)$.
\end{definition}

Since the length of $x$ fully determines the set of solutions, the relevant search problem here is: ``given $1^n$, find some $y$ such that $R(1^n,y)$ holds.'' All explicit construction problems considered in Section~\ref{sec:constructions-in-apepp} will be in ${\cl STF\Sigma_2^P}$ (with the exception of {\sc Complexity} which we briefly mention as it was studied previously in \cite{total}).

We now define the search problem {\sc Empty}, which will be the primary subject of this work:
\begin{definition}
{\sc Empty} is the following search problem: given a boolean circuit $C$ with $n$ input wires and $m$ output wires where $m>n$, find an $m$-bit string outside the range of $C$.
\end{definition}

This problem is total due to the basic lemma, referred to in \cite{total} as the ``Empty Pigeonhole Principle'' and in the field of Bounded Arithmetic as the ``Dual Pigeonhole Principle \cite{Jerabek},'' which tells us that a map from a smaller set onto a larger one cannot be surjective. Since verifying a solution $y$ consists of determining that for all $x$, $C(x) \neq y$, we have:
\begin{observation}
{\sc Empty} $\in$ ${\cl TF\Sigma^{\cl P}_2}$
\end{observation}

Since for any instance of {\sc Empty} the number of output bits $m$ is at least $n+1$, a random $m$-bit string will be a solution with probability at least $\frac{1}{2}$. Since verifying a solution can be accomplished with one call to an {\cl NP} oracle, this implies the following inclusion:

\begin{observation}
{\sc Empty} $\in$ ${\cl FZPP}^{\cl NP}$
\end{observation}

As mentioned in the introduction, this fact tells us that sufficiently strong pseudorandom generators capable of fooling \emph{nondeterministic} circuits such as those in \cite{AM-derandomization} would suffice to derandomize the above inclusion and yield a ${\cl P}^{\cl NP}$ algorithm for {\sc Empty}. In Section~\ref{sec:completeness}, we will show that this derandomization can be accomplished under a significantly weaker assumption, using a reduction of a very different form then the hardness-based pseudorandom generators of \cite{NW}, \cite{IW}, and \cite{AM-derandomization}.

We can now define the class {\cl APEPP}, which is simply the class of search problems polynomial-time reducible to {\sc Empty}. This class was originally defined in \cite{total}, and is an abbreviation for ``Abundant Polynomial Empty Pigeonhole Principle.'' The term ``Abundant'' was used to distinguish this from the larger class {\cl PEPP} also studied in \cite{total}. The complete problem for {\cl PEPP} is to find a string outside the range of a map $C:\{0,1\}^n \setminus \{0^n\} \rightarrow \{0,1\}^n$, which appears significantly more difficult (it is at least as hard as {\cl NP} \cite{total}). The distinction between {\cl APEPP} and {\cl PEPP} also appears in the Bounded Arithmetic literature, where the principle corresponding to {\cl APEPP} is referred to as the ``Dual \emph{weak} Pigeonhole Principle,'' while the principle corresponding to {\cl PEPP} is referred to simply as the ``Dual Pigeonhole Principle.'' We will be concerned only with the abundant/weak principle in this work. It should be noted that we employ a slight change of notation from \cite{total} for the sake of simplicity: we use {\sc Empty} to refer to the search problem associated with the \emph{weak} pigeonhole principle, while in \cite{total} {\sc Empty} refers to the search problem associated with the full pigeonhole principle.

\paragraph{Infinitely-often vs. almost-everywhere circuit lower bounds:} As a final point of clarification, whenever we make the statement ``$L$ requires circuits of size $s(n)$'' for some language $L$ and size bound $s$, we mean that circuits of size $s(n)$ are required to compute $L$ on length $n$ inputs for \emph{all but finitely many $n$}. This is in contrast to the statement ``$L \notin {\cl SIZE}(s(n))$,'' which means the circuit size lower bound holds for infinitely many input lengths. All circuit lower bounds referred to in this work will be of the first kind.

\section{Explicit Constructions in {\cl APEPP}}\label{sec:constructions-in-apepp}

In this section, we show that a variety of well-studied explicit construction problems can be reduced in polynomial time to {\sc Empty}. Each proof follows roughly the following format: there is some property of interest $\Pi$, and our goal is to construct an $n$-bit string which holds this property. For each such $\Pi$ we consider, whenever an $n$-bit string $x$ fails to have this property, it indicates that $x$ is somehow more ``structured'' than a random $n$-bit string, and this structure allows us to specify $x$ using fewer then $n$ bits. We then actualize this argument in the form of an efficiently computable map $C: \{0,1\}^k \rightarrow \{0,1\}^n$ with $k < n$, such that any string not having property $\Pi$ is in the range of $C$. This immediately implies that any $n$-bit string outside the range of $C$ must hold property $\Pi$, and thus any solution to the instance of {\sc Empty} defined by $C$ will be a solution to our explicit construction problem. For many of the proofs, we will only show that the reduction is valid for $n$ sufficiently large; clearly this is sufficient, since explicit constructions can be done by brute force for fixed input lengths.

\paragraph{A useful coding lemma:} In the proofs to come, it will be helpful to utilize succinct and efficiently computable encodings of low-weight strings (the ``weight'' of binary string is the number of 1 bits it contains). We start with the following folklore result, reproduced in \cite{Coin-Problem}, which gives an optimal encoding of weight-$k$ $n$-bit strings:

\begin{lemma}\label{lem:subset-enc}\cite{Coin-Problem}
For any $k \leq n$, there exists a map $\Phi: \{0,1\}^{\log\binom{n}{k}} \rightarrow \{0,1\}^n$ computable in $poly(n)$ time such that any $n$-bit string of weight $k$ is in the range of $\Phi$.
\end{lemma}
As a useful corollary we get the following:
\begin{lemma}\label{lem:sparse-enc}
For any $0 < \epsilon < \frac{1}{2}$, there exists a map  $\Phi: \{0,1\}^{n - \epsilon^2n + \log n} \rightarrow \{0,1\}^n$ computable in $poly(n)$ time such that any $n$-bit string of weight at most $\frac{n}{2} - \epsilon n$ is in the range of $\Phi$.
\end{lemma}
\begin{proof}
By the previous lemma, we are able to efficiently encode $n$-bit strings of weight exactly $k$ for any $k \leq \frac{n}{2} - \epsilon n$ using at most $\log\binom{n}{\frac{n}{2} - \epsilon n}$ bits. We can upper bound $\log\binom{n}{\frac{n}{2} - \epsilon n}$ as follows. Letting $X$ denote the sum of $n$ independent unbiased variables over $\{0,1\}$, we have:
\[
\binom{n}{\frac{n}{2} - \epsilon n} \leq 2^n Pr[X \leq \frac{n}{2} - \epsilon n]
\]
Using a standard Chernoff bound we have that for any $\delta \in (0,1)$:
\[
Pr[X \leq (1-\delta)\frac{n}{2}] \leq \exp(-n\delta^2/4) \leq 2^{-n\delta^2/4}
\]
So setting $\delta = 2\epsilon$ we get $\log\binom{n}{\frac{n}{2} - \epsilon n} \leq n - \epsilon^2n$.

For a string of weight \emph{at most} $\frac{n}{2} - \epsilon n$, we can append an additional $\log n$ bits specifying the weight $k$ of our string, together with the $n - \epsilon^2n$ bits needed to specify a string of weight exactly $k$, to get the desired result. 
\end{proof}


\subsection{Hard Truth Tables}
\begin{definition}\label{def:circ-size}
Given a string $x$ of length $N$, we say that $x$ is computed by a circuit of size $s$ if there is a boolean circuit $C$ of fan-in 2 over the basis $\langle \land, \lor, \lnot \rangle$ with $\lceil \log N \rceil$ inputs and $s$ gates, such that $C(i)=x_i$ for all $1 \leq i \leq |x|$. If $N$ is not a power of 2, we put no restriction on the value of $C(i)$ for $i > |x|$.
\end{definition}

\begin{definition}
{\sc Hard Truth Table} is the following search problem: given $1^N$, output a string $x$ of length $N$ such that $x$ is not computed by any circuit of size at most $\ttbound$.
\end{definition}
In the typical case where $N = 2^n$ for some $n$, this is equivalent to finding a truth table for an $n$-input boolean function requiring circuits of size $\varbound$, which is within a $2 + o(1)$ factor of the worst case circuit complexity for any $n$-input boolean function.

\begin{theorem}\label{thm:truth-tables-in-apepp}
{\sc Hard Truth Table} reduces in polynomial time to {\sc Empty}.
\end{theorem}
\begin{proof}
This proof follows Shannon's classical argument for the existence of functions of high circuit complexity \cite{Shannon}. We construct an instance of {\sc Empty} in the form of a circuit $\Phi$ which maps an encoding of a circuit to its corresponding truth table. $\Phi$ interprets its input as a circuit on $\lceil \log N \rceil$ bits (using an encoding of circuits to be described below), tests its value on every possible input to generate a $2^{\lceil \log N \rceil}$ bit truth table, and then truncates this truth table to be of length exactly $N$.

We now describe the encoding of circuits used by $\Phi$, which will guarantee that any $N$-bit string $x$ with a circuit of size $\frac{N}{2\log N}$ is in the range of $\Phi$. Given a circuit of size $s$ on $\lceil \log N \rceil$ inputs computing $x$, for each of its $s$ gates we can use 2 bits to encode whether it is an $\land, \lor$, or $\lnot$ gate, and an additional $2\log s$ bits to specify its inputs. We can then use an additional $\log s$ bits to specify which gate is the terminal output gate. Overall this requires $2s\log s + O(s)$ bits. It is clear that from such an encoding, $\Phi$ can efficiently decode the represented circuit and test it on all possible input values. For $s \leq \frac{N}{2\log N}$, we have:
\[
2s\log s + O(s) \leq \frac{N}{\log N}\log\left(\frac{N}{2\log N}\right) + O\left(\frac{N}{\log N}\right) = N - \Omega\left(\frac{N\log\log N}{\log N}\right) + O\left(\frac{N}{\log N}\right)
\]

which is strictly less then $N$ for $N$ sufficiently large. So $\Phi$ is indeed a valid instance of {\sc Empty}, and any string outside the range of $\Phi$ is a solution to {\sc Hard Truth Table}. It is also clear from the above description that $\Phi$ can be constructed in $poly(N)$ time.

\end{proof}

A related problem was studied in \cite{total}, referred to there as ``{\sc Complexity}.'' Rather then an explicit construction problem with sparse input, in this problem you are given as input one truth table $x$ of length $N$, and asked to produce another truth table $y$ such that $y$ requires large circuits, even with access to $x$-oracle gates. More formally:

\begin{definition}
The problem {\sc Complexity} is defined as follows: given an $N$-bit string $x$, find another $N$-bit string $y$ such that $y$ requires $x$-oracle circuits of size $\frac{N}{\log^2 N}$.
\end{definition}
Note that the ``$x$ oracle'' is not the typical definition of an oracle gate that can solve arbitrarily sized instances of a fixed language, but rather an oracle for a fixed boolean function on $\log N$ variables. 
In \cite{total} the following is shown:
\begin{theorem}
{\sc Complexity} reduces in polynomial time to {\sc Empty}.
\end{theorem}
It should be noted that {\sc Complexity} bears resemblance to a problem studied by Ilango \cite{Ilango}, termed the ``Minimum Oracle Circuit Size Problem,'' or ``MOCSP.'' In this problem, the input consists of two truth tables $x$ and $y$ and a size parameter $s$, and the goal is to determine if $y$ has $x$-oracle circuits of size at most $s$. Ilango demonstrates that MOCSP is {\cl NP}-complete under randomized reductions.

\subsection{Pseudorandom Generators}
\begin{definition}
We will say that a sequence $R = (x_1, \ldots, x_m)$ of $n$-bit strings is a pseudorandom generator if, for all $n$-input circuits of size $n$:
\[
|Pr_{x \sim R}[C(x)=1] - Pr_{y \sim \{0,1\}^n}[C(y)=1]| \leq 1/n
\]
\end{definition}

Standard applications of the probabilistic method show that such pseudorandom generators exist of size polynomial in $n$. Thus we can define the following total search problem:
\begin{definition}
{\sc PRG} is the following search problem: given $1^n$, output a pseudorandom generator $R = (x_1, \ldots, x_m)$, $x_i \in \{0,1\}^n$.
\end{definition}

A polynomial time algorithm for {\sc PRG} would suffice to derandomize {\cl BPP} \cite{NW}. We now show how to formalize the argument for the totality of {\sc PRG} using the empty pigeonhole principle. In particular, we show that a PRG of size $n^6$ can be constructed in {\cl APEPP}.

As noted in the introduction, the results of Impagliazzo and Wigderson \cite{IW} imply that {\sc PRG} reduces directly to {\sc Hard Truth Table}, so a reduction of {\sc PRG} to {\sc Empty} follows from Theorem~\ref{thm:truth-tables-in-apepp}. However, we provide here a much simpler direct proof that {\sc PRG} reduces to {\sc Empty}, relying only on Yao's next bit predictor lemma, and neither the nearly disjoint subsets construction of Nisan and Wigderson \cite{NW} nor the rather involved worst-case to average-case reductions of Impagliazzo and Wigderson \cite{IW}. Together with our completeness result in Section~\ref{sec:completeness}, this gives an alternative, self-contained proof that worst-case-hard truth tables can be used to construct pseudorandom generators (although it yields a weaker result, as our derandomization will require an {\cl NP} oracle).

\begin{theorem}\label{thm:prg}
{\sc PRG} reduces in polynomial time to {\sc EMPTY}
\end{theorem}

\begin{proof}
We start by constructing the following circuit $\Phi$. $\Phi$ interprets its input as representing a list $R^- = (x_1, \ldots, x_{n^6})$ of $n^6$ strings of length $n-1$, a circuit $D$ of size $cn$ (for a fixed universal constant $c$ to be defined later) with $n-1$ input bits and one output bit, a $\log n$-bit index $i \in [n]$, and a string $S$ encoding an $n^6$-bit string with weight at most $\frac{n^6}{2} - n^4$. Given these, $\Phi$ feeds each $n-1$-bit string $x_j \in R^-$ through $D$, then inserts the output as an extra bit at the $i^{th}$ position of $x_j$ to obtain an $n$-bit string $x_j^*$. In the final step, $\Phi$ decodes the $n^6$-bit string represented by $S$, and flips $i^{th}$ bit of $x_j^*$ if and only if the $j^{th}$ position of $S$ has a 1, to obtain $x_j'$. $\Phi$ then outputs $R = (x_1', \ldots, x_{n^6}')$.

It is clear that $\Phi$ can be implemented as a circuit of size polynomial in $n$, and further that $\Phi$ can be constructed in polynomial time given $1^n$. We now show there are fewer inputs then outputs. Note that the input size is equal to the number of bits needed to specify $R^-, D, i, S$, which is $n^7 - n^6 + \Tilde{O}(n) + \log n + bits(S)$, where $bits(S)$ is the number of bits needed to specify $S$. Since $S$ an $n^6$-bit string with weight at most $\frac{n^6}{2} - n^4 = n^6(\frac{1}{2} - \frac{1}{n^2})$, we can apply Lemma~\ref{lem:sparse-enc} with $\epsilon=\frac{1}{n^2}$ to give an encoding for $S$ using $n^6(1-\frac{1}{n^4}) + \log (n^6) = n^6 - n^2 + 6\log n$ bits. Thus the overall number of bits of input is at most $n^7 - n^6 + \Tilde{O}(n) + \log n + n^6 - n^2 + 6\log n = n^7 + \Tilde{O}(n) + 7\log n - n^2$ which is strictly less then the $n^7$ bits of output (for sufficiently large $n$). So this is indeed a polynomial time reduction to a valid instance of {\sc Empty}. It remains to show that any string outside the range of $\Phi$ is a pseudorandom generator.

Let $R$ be a sequence of $n$-bit strings of size $n^6$ which is not a pseudorandom generator. So there exists a circuit $C$ of size $n$ such that:
\[
|Pr_{x \sim R}[C(x)=1] - Pr_{y \sim \{0,1\}^n}[C(y)=1]| > 1/n
\]

For $x \in \{0,1\}^n$, $i \in [n]$, let $x^{-i}$ be the $n-1$-bit string obtained by deleting the $i^{th}$ bit of $x$, and let $x^i$ denote the $i^{th}$ bit of $x$. By Yao's next bit predictor lemma \cite{Yao} \cite{Vadhan}, the previous inequality implies the existence of an index $i \in [n]$ and a circuit $D: \{0,1\}^{n-1} \rightarrow \{0,1\}$ of size $cn$ for some fixed universal constant $c$, such that 
\[
Pr_{x \sim R}[D(x^{-i}) = x^i] > \frac{1}{2} + \frac{1}{n^2}
\]
Since $R$ has size $n^6$, this implies that $D$ correctly guesses the $i^{th}$ bit of $x \in R$ from the other $n-1$ bits for at least $n^6(\frac{1}{2} + \frac{1}{n^2}) = \frac{n^6}{2} + n^4$ of the elements of $R$. So if we let $S$ be the $n^6$-bit string with a 1 at the indices where $D$ guesses the wrong value of $x^i$, we see that the weight of $S$ is at most $\frac{n^6}{2} - n^4$ as required. Taking $R^- = (x_1^{-i}, \ldots, x_{n^6}^{-i})$, we have that from $R^-,D,i,S$ we can efficiently deduce $R$.

So overall, have established that for any $R$ which is not a pseudorandom generator, there exists some $R^-,D,i,S$ such that our circuit $\Phi$ outputs $R$ on input $R^-,D,i,S$. Thus, any string outside the range of $\Phi$ is a pseudorandom generator.
\end{proof}

\subsection{Strongly Explicit Randomness Extractors and Ramsey Graphs}

A $(k,\epsilon)$ two-source extractor with one bit of output is a function $f: \{0,1\}^n \times \{0,1\}^n \rightarrow \{0,1\}$ such that for any pair of distributions $X,Y$ on $\{0,1\}^n$ of min-entropy at least $k$, the value of $f(xy)$ for a random $x,y$ in the product distribution of $X,Y$ is $\epsilon$-close to an unbiased coin flip. By a well-known simplification of \cite{flat-sources}, to show that a function $f$ is a $(k,\epsilon)$ extractor, it suffices to show that it satisfies the above condition for every pair of ``flat'' $k$-sources $X,Y$, which are uniform distributions over subsets of $\{0,1\}^n$ of size $2^k$. We will thus use the following definition of two-source extractors to define our explicit construction problem:
\begin{definition}
We say that a function $f:\{0,1\}^n \times \{0,1\}^n \rightarrow \{0,1\}$ is a $(k,\epsilon)$ extractor if the following holds: for any two sets $X,Y \subseteq \{0,1\}^n$ of size $2^k$, $|Pr_{x \sim X, y \sim Y}[f(xy) = 1] - \frac{1}{2}| \leq \epsilon$.
\end{definition}

\begin{definition}
For any pair of functions $k, \epsilon: \mathbb{N} \rightarrow \mathbb{N}$, {\sc $(k,\epsilon)$-Extractor} is the following search problem: given $1^n$, output a circuit $C$ with $2n$ inputs such that the function $f_C: \{0,1\}^n \times \{0,1\}^n \rightarrow \{0,1\}$ defined by $C$ is a $(k(n),\epsilon(n))$ extractor.
\end{definition}

The above problem definition does not expressly constrain the size of $C$, though for a construction to be ``explicit'' in any useful sense (efficiently computable as a function of $n$), $C$ would have to have size polynomial in $n$. The following reduction placing extractor construction in {\cl APEPP} will immediately imply that we can construct a $(\log n + O(1), \epsilon)$ extractor of circuit size approximately $n^3$ in {\cl APEPP} for any fixed $\epsilon$.

\begin{theorem}\label{thm:ext}
For any efficiently computable $\epsilon(n)$ satisfying $\frac{1}{n^c} < \epsilon(n) < \frac{1}{2}$ for a constant $c$ and sufficiently large $n$, {\sc $(\log n + 2\log(1/\epsilon(n)) + 3, \epsilon(n))$-Extractor} reduces in polynomial time to {\sc Empty}.
\end{theorem}
\begin{proof}
Let $\epsilon = \epsilon(n)$, and let $d = \lceil \frac{4}{\epsilon^2} \rceil$. 
We will set up an instance of {\sc Empty} with at most $2d^2n^3 + 2dn^2 - \epsilon^2d^2n^2 + 2\log(2dn) + 1$ inputs and exactly $2d^2n^3$ outputs which has the following property: Let $A$ be any $2d^2n^3$-bit string outside the range of this circuit, viewed as an ordered list of $d^2n^2$ elements of $\mathbb{F}_{2^{2n}}$ denoted $\alpha_1, \ldots \alpha_{d^2n^2}$, and consider the function $f: \{0,1\}^{2n} \rightarrow \{0,1\}^{2n}$ defined by
\[
f(x) = \sum\limits_{i=1}^{d^2n^2}\alpha_i x^{i-1}
\]
Then the function $g: \{0,1\}^{2n} \rightarrow \{0,1\}$ defined by 
\[
g(x) = f(x) \mod 2
\]
is a $(\log dn, \epsilon)$ extractor. Since $\log dn = \log n + \log d \leq \log n + \log(4/\epsilon^2 + 1) \leq \log n + 2\log(1/\epsilon) + 3$, this would give the required result.

Let $\alpha_1, \ldots \alpha_{d^2n^2}$ be any sequence of coefficients in $\mathbb{F}_{2^{2n}}$ such that the above function $g$ corresponding to the $\alpha_i$ is not a $(\log dn, \epsilon)$ extractor. So there exist two sets of $n$-bit strings $X,Y$, each of size $2^{\log dn} = dn$, and some $b \in \{0,1\}$ such that $Pr_{x \sim X, y \sim Y}[f(xy) = b] > \frac{1}{2} + \epsilon$. Let $R = \{xy \mid x \in X, y \in Y\} \subseteq \{0,1\}^{2n}$. We have $|R| = |X||Y| = d^2n^2$. Let $r_1, \ldots r_{d^2n^2}$ denote the lexicographical enumeration of $R$. By assumption, we have that $g(r_i) = b \mod 2$ for at least a $\frac{1}{2} + \epsilon$ fraction of indices $i$. So then, if we let $\beta_i$ be the $2n-1$-bit prefix of $f(r_i)$, we can deduce the value of $f(r_i)$ from $\beta_i$ and $b$ for at least $d^2n^2(\frac{1}{2} + \epsilon)$ values of $i$. Thus, there is some $d^2n^2(\frac{1}{2} - \epsilon)$-weight $d^2n^2$-bit string $S$, such that from $b$, $S$, and the $\beta_i$'s, we can deduce $f(r_i)$ for all $i$. Now, once we are able to deduce $f(x)$ for each of the $d^2n^2$ distinct values of $x$ in $R$, since $f$ is a degree $d^2n^2-1$ polynomial, we can uniquely and efficiently determine the coefficients $\alpha_i$ of $f$ using Gaussian elimination on the corresponding $d^2n^2 \times d^2n^2$ Vandermonde matrix.

So we have shown that for any set of $\alpha_i$'s which does not define a $(\log dn, \epsilon)$ extractor, there exists $X,Y,b,S$, and $\beta_i$'s, from which we can efficiently deduce the $\alpha_i$'s. It is clear that we can encode $X,Y$ using $2dn^2$ bits, $b$ using 1 bit, and the $\beta_i$'s using $d^2n^2(2n-1) = 2d^2n^3 - d^2n^2$ bits. Since $S$ is a $d^2n^2(\frac{1}{2} - \epsilon)$-weight $d^2n^2$-bit string, by Lemma~\ref{lem:sparse-enc}, we can encode $S$ using at most $d^2n^2(1 - \epsilon^2) + \log(d^2n^2) = d^2n^2 - \epsilon^2d^2n^2 + \log(d^2n^2)$ bits. So in total the number of input bits is at most:
\[
2dn^2 + 1 + d^2n^2 - \epsilon^2d^2n^2 + \log(d^2n^2) + 2d^2n^3 - d^2n^2 = 2d^2n^3 + (2d-\epsilon^2d^2)n^2 + 2\log(dn) + 1
\]
Since we chose $\frac{4}{\epsilon^2} + 1 \geq d \geq \frac{4}{\epsilon^2}$, we have that$(2d - \epsilon^2d^2) < -1$ for $\epsilon \leq \frac{1}{2}$, so the number of input bits is at most $2d^2n^3 - n^2 + 2\log(dn) + 1$. Since we assumed $\epsilon \geq \frac{1}{n^c}$ for a constant $c$, we have $\epsilon \geq \frac{1}{2^{n^2/5}}$ for sufficiently large $n$, and so $n^2 > 2\log(dn) + 1$, which is strictly less then the number of output bits $2d^2n^3$.

So this implies that the circuit $\Phi$ mapping $X,Y,b,S$ and the $\beta_i$'s to a corresponding set of $\alpha_i$'s is a valid instance of {\sc Empty}, and that any solution to this instance is a set of coefficients defining a $(\log n + 2\log(1/\epsilon) + 3, \epsilon)$ extractor. From the coefficients $\alpha_i$ we can easily construct an efficient circuit computing the function $g   $, thus solving the problem {\sc $(\log n + 2\log(1/\epsilon) + 3, \epsilon)$-Extractor}. It is also clear from the construction that this reduction can be carried out in $poly(n, \frac{1}{\epsilon})$ time, so $poly(n)$ time under our assumption that $\epsilon(n) \geq \frac{1}{n^c}$ for a constant $c$.
\end{proof}

For the typical parameter regime where $\epsilon$ is an arbitrarily small constant, this gives a two source extractor for min-entropy $\log n + O(1)$.

\begin{corollary}
Explicit construction of strongly explicit Ramsey graphs ($n$-vertex graphs containing no clique or independent set of size $c\log n$ for some constant $c$), in both the bipartite and non-bipartite case, reduces to {\sc Empty}.
\end{corollary}

\begin{proof}
As noted in \cite{BRSW}, any two-source extractor in the above sense (with $\epsilon$ fixed to any constant less then one half) is automatically a bipartite Ramsey graph, and from a strongly explicit bipartite Ramsey graph we can construct a strongly explicit non-bipartite Ramsey graph efficiently.
\end{proof}

\subsection{Rigid Matrices}
\begin{definition}\cite{Valiant}
We say that $n \times n$ matrix $M$ over $\mathbb{F}_q$ is $(r, s)$ rigid if for any matrix $S \in \mathbb{F}_q^{n \times n}$ with at most $s$ non-zero entries, $M+S$ has rank greater than $r$.
\end{definition}

\begin{definition}
For any $q: \mathbb{N} \rightarrow \mathbb{N}$ such that $q(n)$ is a prime power $\forall n$, {\sc $(\epsilon, q)$-Rigid} is the following search problem: given $1^n$, output an $n \times n$ matrix $M$ over $\mathbb{F}_{q(n)}$ which is $(\epsilon n, \epsilon n^2)$ rigid.
\end{definition}

\begin{theorem}\label{thm:rigid-in-apepp}
For any $\epsilon \leq \frac{1}{16}$, and any efficiently computable $q(n)$ satisfying the above, {\sc $(\epsilon, q)$-Rigid} reduces in polynomial time to {\sc Empty}.
\end{theorem}

\begin{proof}
Let $M$ be any $n \times n$ matrix over $\mathbb{F}_q$ which is not $(r,s)$ rigid. So there exists an $n \times r$ matrix $L$, an $r \times n$ matrix $R$, and an $n \times n$ matrix $S$ with at most $s$ non-zero entries, such that $M = LR + S$. It is clear that from the descriptions of $L,R,S$ we can efficiently compute $M$.

$L$ and $R$ can each be described explicitly using $nr \log q$ bits. For $S$, we encode it by specifying an $n^2$-bit string $T$ of weight $s$ denoting the entries of $S$ which are nonzero, together with an $s \log q$-bit string giving the values of the nonzero entries. Applying the encoding scheme in Lemma~\ref{lem:subset-enc} for $T$, overall the number of bits in this encoding is at most $\log\binom{n^2}{s} + (2nr + s)\log q$. Setting $r = \epsilon n$ and $s= \epsilon n^2$ this is at most:
\begin{gather*}
\log\binom{n^2}{\epsilon n} + (3\epsilon n^2)\log q
 \leq \\
(\frac{3}{4} + \epsilon - \epsilon^2)n^2 + (3 \epsilon n^2)\log q \leq \\
(\frac{3}{4} + \epsilon - \epsilon^2 + 3\epsilon)n^2\log q
\end{gather*}
Since we chose $\epsilon \leq \frac{1}{16}$, this is at most $\frac{997}{1000}n^2\log q$, which is strictly less then the $n^2\log q$ bits needed to specify an arbitrary matrix in $\mathbb{F}_q^{n \times n}$.
\end{proof}

\subsection{Strings of High Time-Bounded Kolmogorov Complexity}

\begin{definition}
Let $U$ be any fixed Turing machine, and let $t: \mathbb{N} \rightarrow \mathbb{N}$ be a time bound. For a string $x$, $K_U^t(x)$ denotes the length of the smallest string $y$ such that $U$ outputs $x$ on input $y$ in $t(|x|)$ steps. 
\end{definition}

\begin{definition}
For a Turing machine $U$ and time bound $t$, we define the following explicit construction problem {\sc $K_U^t$-Random}: given $1^n$, output an $n$ bit string $x$ such that $K_U^t(x) \geq n - 1$
\end{definition}

\begin{theorem}\label{thm:kt-random}
For any fixed Turing machine $U$ and fixed polynomial time bound $t$, {\sc $K_U^t$-Random} reduces in polynomial time to {\sc Empty}.
\end{theorem}
\begin{proof}
We construct a circuit $\Phi$ with $n-1$ input bits and $n$ output bits as follows. If the input is of the form $0^*1y$, $\Phi$ simulates $U$ on $y$ for $t(n)$ steps and outputs the result, padded/truncated to length $n$ if it is not already. Otherwise, $\Phi$ outputs $0^n$. It is clear that $\Phi$ can be produced in polynomial time given $1^n$ for any fixed $U,t$. Further, for any $n$-bit string $x$ with $K_U^t(x) \leq n - 2$, there is some $y$ of length $k \leq n-2$ such that $U$ outputs $x$ on input $y$ in $t(n)$ steps, and so $\Phi$ will output $x$ on input $0^{n-k-2}1y$. So any string $x$ outside the range of $\Phi$ must have $K_U^t(x) \geq n - 1$.
\end{proof}

\subsection{Other Problems}
In Section~\ref{sec:direct-reductions}, we introduce two more explicit construction problems and show that each of these, in addition to a variant of the rigidity problem, can be reduced directly to {\sc Hard Truth Table} in polynomial time. This also implies that both problems are contained in {\cl APEPP}. We will  postpone a formal definition of each of these new problems until Section~\ref{sec:direct-reductions}, but give an informal statement here for completeness:

\paragraph{Explicit communication problems outside ${\cl PSPACE^{CC}}$:} An explicit $2^n \times 2^n$ communication matrix which cannot by solved by any $o(n)$-space protocol can be constructed in  {\cl APEPP}.

\paragraph{Explicit data structure problems with high bit-probe complexity:} An explicit data structure problem with nearly maximum complexity in the bit-probe model can be constructed in {\cl APEPP}.


\section{Constructing Hard Truth Tables is Complete for {\cl APEPP}}\label{sec:completeness}

In this section we show that constructing a hard truth table is complete for {\cl APEPP} under ${\cl P}^{\cl NP}$ reductions. As mentioned before, the core of this theorem was originally proven by Je\v{r}\'{a}bek \cite{Jerabek}, and the main construction underlying the reduction dates back further to the work of Goldreich, Goldwasser, and Micali \cite{GGM}.  Je\v{r}\'{a}bek's result is phrased in the language of proof complexity, stating that the theorem asserting the existence of hard boolean functions is \emph{equivalent} to the empty pigeonhole principle in a particular theory of Bounded Arithmetic. We demonstrate below that when translated to the language of search problems and explicit constructions, his proof yields a ${\cl P}^{\cl NP}$ reduction from {\sc Empty} to the problem of constructing a hard truth table. We in fact prove a more general statement here which holds for arbitrary circuit classes equipped with oracle gates. For a very broad set of circuit classes $\mathcal{C}$, we prove that given a truth table which is hard for $\mathcal{C}$-circuits, we can find an empty pigeonhole of any $\mathcal{C}$-circuit \emph{using polynomially many calls to an oracle for inverting $\mathcal{C}$-circuits} (by inverting we mean finding the preimage of a given string, or reporting that none exist). In order to make this precise, we first define the type of generalized circuit classes we will consider:

\begin{definition}
A circuit class is defined by a basis $\mathcal{C}$, which is simply a (possibly infinite) set of boolean-valued boolean functions. A $\mathcal{C}$-circuit is then defined as a circuit composed entirely of gates computing functions in the basis $\mathcal{C}$. For a language $L$, we will refer to the ``basis $L$'' to mean the basis $\{L_n \mid n \in \mathbb{N}\} \cup \{\land,\lor,\lnot\}$, where $L_n$ is the $n$-bit boolean function deciding $L$ on length $n$ inputs. For a complexity class {\cl C} with a complete language $L$, we will refer to ``the basis {\cl C}'' to mean the basis corresponding to $L$.

We say that a basis $\mathcal{C}$ is ``sufficiently strong'' if there exist $\mathcal{C}$-circuits computing the two-input $\land, \lor$ functions and the one-input $\lnot$ function.

For any basis $\mathcal{C}$, a ``$\mathcal{C}$-inverter oracle'' is an oracle which, given a $\mathcal{C}$-circuit $C$ and some string $y$, determines whether there exists an $x$ such that $C(x)=y$, and produces such an $x$ if it exists. A $\mathcal{C}$-inverter reduction is a polynomial time reduction that uses a $\mathcal{C}$-inverter oracle.
\end{definition}

Note that in the special case where $\mathcal{C} = \{\land, \lor, \lnot\}$, a $\mathcal{C}$-inverter reduction is equivalent to a ${\cl P}^{\cl NP}$ reduction (since inverting a circuit over the standard boolean basis is {\cl NP}-complete).

\begin{definition}
For any basis $\mathcal{C}$, the class of search problems {\cl APEPP$^{\mathcal{C}}$} is defined by the following complete problem {\sc EMPTY$^{\mathcal{C}}$}: given a $\mathcal{C}$-circuit with more output wires than input wires, find a boolean string whose length is equal to the number of output wires but which is not in the range of this circuit. For any strictly increasing function $f: \mathbb{N} \rightarrow \mathbb{N}$, we define the problem {\sc $EMPTY^{\mathcal{C}}_{f(n)}$}, which is the special case of {\sc EMPTY$^{\mathcal{C}}$} where the circuit is required to have $f(n)$ output wires, where $n$ is the number of input wires.
\end{definition}

We start with the following technical lemma, which allows us to restrict our attention to circuits with exactly twice as many outputs as inputs.

\begin{lemma}\label{lem:swapping-sizes}
For any basis $\mathcal{C}$, {\sc $EMPTY_{2n}^{\mathcal{C}}$} is complete for ${\cl APEPP^{\mathcal{C}}}$ under $\mathcal{C}$-inverter reductions.
\end{lemma}
\begin{proof}
It is straightforward to see that {\sc $EMPTY_{n+1}^{\mathcal{C}}$} is complete for {\cl APEPP}$^{\mathcal{C}}$: given a $\mathcal{C}$-circuit $C$ with $n$ inputs and more than $n$ outputs, we can simply delete all the output bits except for the first $n+1$ to obtain an instant of {\sc $EMPTY_{n+1}^{\mathcal{C}}$}. Any $n+1$-bit string outside the range of this smaller circuit can than be padded arbitrarily to the output size of the original circuit, and this padded string must also be outside the range of $C$.

We now reduce {\sc $EMPTY_{n+1}^{\mathcal{C}}$} to {\sc $EMPTY_{2n}^{\mathcal{C}}$}. Let $C$ be some $\mathcal{C}$-circuit with $n$ inputs and $n+1$ outputs. For a positive integer $i$, we define $C^i: \{0,1\}^n \rightarrow \{0,1\}^{n + i}$ inductively as follows. For the base case, $C^1$ is simply defined as $C$. For $i > 1$, we first compute $C^{i-1}$ on the input $x$ to obtain a string $x'$ of length $n + i - 1$, and then compute $C$ on the first $n$ bits of $x'$ and concatenate the output with the remaining bits of $x'$. Since $C$ replaces the first $n$ bits with $n+1$ bits, if $C^{i-1}$ has output length $n + i-1$ then $C^i$ has output length $n + i$.

We now claim that for any positive $i$, if we are given an $n+i$-bit string outside the range of $C^i$, we can find an $n$-bit string outside the range of $C$ in $poly(i|C|)$ time using a $\mathcal{C}$-inverter. In particular, we prove by induction on $i$ that this can be accomplished with $i$ inverter calls. For $i=1$ this is trivially the case, since $C^1=C$. Now say $i > 1$. Let $y \in xz$ be outside the range of $C^i$ for some $x \in \{0,1\}^{n+1}$, $z \in \{0,1\}^{i-1}$. We first use the inverter oracle to determine if $x$ has a preimage under $C$; if it does not then we have found an empty pigeonhole for $C$ and are done. Otherwise, we use the inverter to find a preimage $x' \in \{0,1\}^n$ for $x$ under $C$. So then $x'z$ must be outside the range of $C^{i-1}$, since if it were not then $xz=y$ would have a preimage under $C^i$, contradicting our initial assumption. So by induction we can use $i-1$ inverter calls to find an empty pigeonhole for $C$ given $x'z$, completing the proof of the inductive case.

Now, setting $i = n$, we get a reduction from {\sc $EMPTY_{n+1}^{\mathcal{C}}$} to {\sc $EMPTY_{2n}^{\mathcal{C}}$} as claimed.

\end{proof}
We now define the hard truth table construction problem that will be used in our reduction:
\begin{definition}
Let {\sc $\epsilon$-Hard}$^{\mathcal{C}}$ denote the following search problem: given $1^N$, output a string $x$ of length $N$ such that $x$ cannot be computed by $\mathcal{C}$-circuits of size $N^\epsilon$.
\end{definition}
In the case where $\mathcal{C} =\{\land,\lor,\lnot\}$, we drop the subscript and refer to this problem simply as {\sc $\epsilon$-Hard}.
For $N = 2^n$, a solution to {\sc $\epsilon$-Hard} on input $1^N$ is a truth table of a function on $n$ variables requiring $2^{\epsilon n}$-sized circuits, the same object used to build the Impagliazzo-Wigderson generator.

\begin{theorem}\label{thm:epsilon-hard-reduction}
Let $\mathcal{C}$ be any sufficiently strong basis and $\epsilon > 0$ be a constant such that {\sc $\epsilon$-Hard}$^{\mathcal{C}}$ is total for sufficiently large input lengths. Then {\sc Empty}$^{\mathcal{C}}$ reduces in polynomial time to {\sc $\epsilon$-Hard}$^{\mathcal{C}}$ under $\mathcal{C}$-inverter reductions.
\end{theorem}

\begin{proof}
By Lemma~\ref{lem:swapping-sizes} we know that {\sc Empty}$^{\mathcal{C}}$ reduces to {\sc EMPTY}$_{2n}^{\mathcal{C}}$ under $\mathcal{C}$-inverter reductions. Now, let $C$ be an instance of {\sc EMPTY}$_{2n}^{\mathcal{C}}$, and let $k = 2 \lceil \log |C| \rceil \lceil \frac{1}{\epsilon}\rceil$. Consider the following map $C^*: \{0,1\}^n \rightarrow \{0,1\}^{2^kn}$, defined informally as follows: given a string $x \in \{0,1\}^n$, apply $C$ once to get 2 $n$-bit strings, then apply $C$ to both of those $n$-bit strings to get four, and continue $k$ times until we have $2^k$ $n$-bit strings, or equivalently a $2^kn$-bit string. This process is illustrated in Figure~\ref{fig:circuit_tree} below:

\begin{figure}[H]
    \centering
    \includegraphics[width=0.7\textwidth]{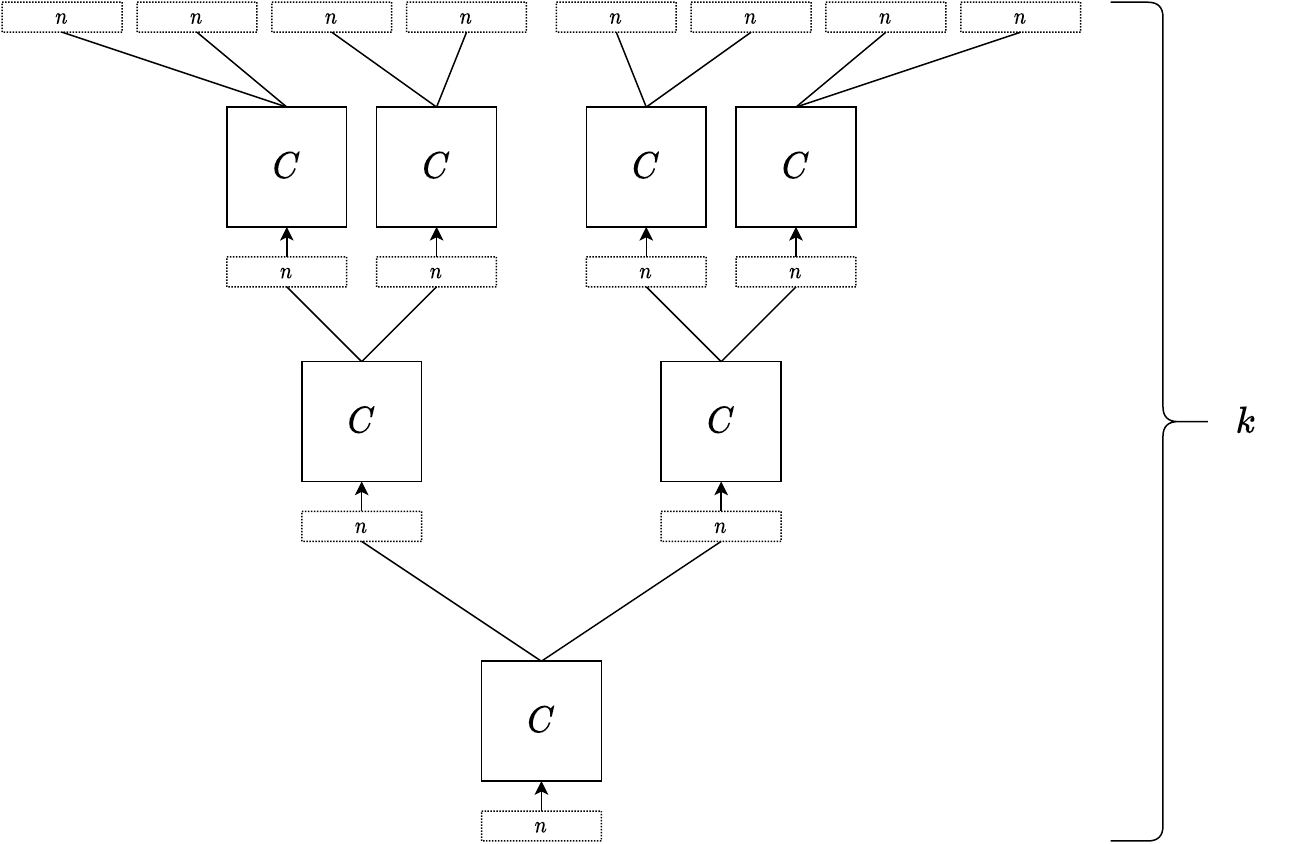}
    \caption{Extending a map $C: \{0,1\}^n \rightarrow \{0,1\}^{2n}$ to a map $C^*: \{0,1\}^n \rightarrow \{0,1\}^{2^kn}$}. Dotted boxes indicate the number of bits along a wire.
    \label{fig:circuit_tree}
\end{figure}

To define this function more formally, first we define the following maps $L,R: \{0,1\}^{2n} \rightarrow \{0,1\}^n$, where $L$ takes a $2n$-bit string and ouputs the first $n$ bits, and $R$ takes a $2n$-bit string and outputs the last $n$ bits. Given a nonempty sequence $\sigma_1, \ldots \sigma_t \in \{L,R\}^*$, let $C^{\sigma}: \{0,1\}^n \rightarrow \{0,1\}^n$ be the function $\sigma_t \circ C \circ \ldots \circ \sigma_1 \circ C$. Now, given a binary string, we can associate it with such a sequence by associating 0 with $L$ and 1 with $R$, and so we will abuse notation and write $C^x$ for a binary string $x$ as shorthand for $C^{\sigma}$ where $\sigma$ is the sequence of $L,R$ associated with the binary string $x$. We are now ready to formally define our function $C^*$. As $C^*$ is a map $\{0,1\}^n \rightarrow \{0,1\}^{2^kn}$, we can think of the output as being clumped into $2^k$ blocks, each containing $n$ bits. Given this terminology, $C^*$ behaves as follows: on input $x$, the $i^{th}$ block of the output of $C^*$ will be $C^{\|i\|}(x)$, where $\|i\|$ denotes the standard representation of $i$ as a $k$-bit binary string.

From here, the proof proceeds in two steps. First, we show that, by setting $m = n 2^k = poly(|C|)$, any solution to {\sc $\epsilon$-Hard}$^{\mathcal{C}}$ on input $1^m$ will be a string that is not in the range of $C^*$. Second, we will show that given a string outside the range of $C^*$, we can find a string outside the range of $C$ using only a polynomial number of calls to a $\mathcal{C}$-inverter.

To carry out the first of these steps, we will show that any string in the range of $C^*$, when interpreted as a truth table of length $m = n 2^k$ on $\lceil \log n \rceil + k$ variables, can be computed by a circuit of size $O(|C|k)$. Since, by construction of $k$, we have that $m \geq |C|^{\frac{2}{\epsilon}}$, a solution to {\sc $\epsilon$-Hard} on input $1^m$ will be a truth table of length $m$ not computable by a circuit of size $m^{\epsilon} \geq |C|^2$, and thus a circuit of size $O(|C|k) = O(|C|\log |C|)$ would be a contradiction for all input lengths greater than some absolute constant. We construct such a circuit for any string in the range of $C^*$ as follows: let $y$ be a $2^kn$-bit string such that for some $x \in \{0,1\}^n$, $C^*(x)=y$. The circuit computing $y$ will have $x$ written as advice/constants, and will feed $x$ through $k$ copies of the circuit $C$ in series. We will split the $\lceil \log n \rceil + k$ input variables into a block of $k$ variables we call $i$, and a block of $\lceil \log n \rceil$ variables we call $j$. We then use $i$ to determine whether to apply $L$ or $R$ to the output of one of the copies of $C$ before feeding it into the next, to get some resulting string $x^i$, and then we use $j$ to index into the $j^{th}$ position of $x^i$, to get $y_{i,j}$. A diagram of this circuit is shown in Figure~\ref{fig:efficient_circuit}:

\begin{figure}[H]
    \centering
    \includegraphics[width=1\textwidth]{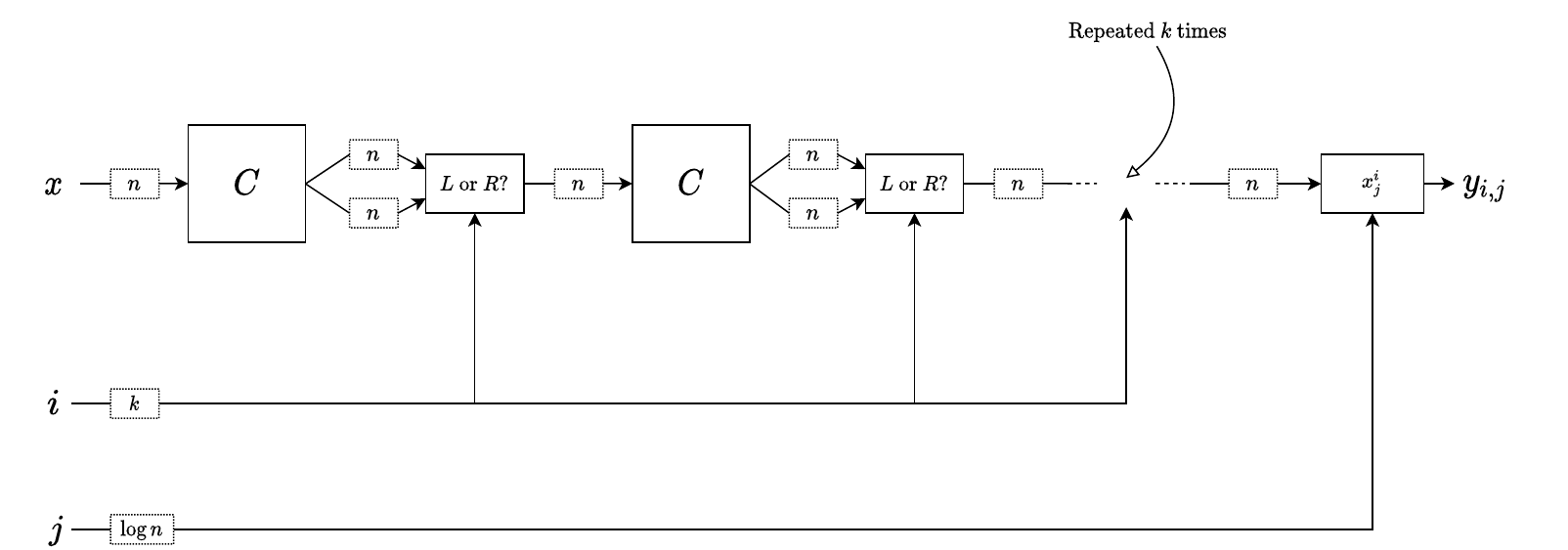}
    \caption{A succinct circuit whose truth table is $y$, for any $y$ in the range of $C^*$. Dotted boxes indicate the number of bits along a wire. Note that although $x$ is shown as an input in this diagram, for any given $y$ we fix a preimage $x$ as constants/advice, and so the only true inputs to this circuit are $i,j$.}
    \label{fig:efficient_circuit}
\end{figure}

To see that this circuit has size $O(|C|k)$, note that the subcircuits computing either $L$ or $R$ depending on a bit of $i$ can be computed easily with $O(n)$ gates over the basis $\{\land, \lor,\lnot\}$ (this is essentially a multiplexer), and also that the final subcircuit indexing into an $n$-bit string can be computed with $O(n)$ $\{\land, \lor,\lnot\}$ gates as well; since we assumed $\mathcal{C}$ is sufficiently strong, both of these subcircuits can therefore be computed with $O(n)$ $\mathcal{C}$-gates. Since $|C| \geq n$, and this circuit contains only $k$ copies of $C$ and the aforementioned subcircuits, plus the constants describing the string $x$ of length $n$, this circuit has $O(|C|k)$ size as claimed. 

 Thus, we now know that any solution to  {\sc $\epsilon$-Hard}$^{\mathcal{C}}$ on input $1^m$ will not be in the range of $C^*$, and by assumption {\sc $\epsilon$-Hard}$^{\mathcal{C}}$ is total for sufficiently large input lengths so such a solution exists. It remains only to show that we can use a string outside the range of $C^*$, together with a $\mathcal{C}$-inverter oracle, to find a string outside the range of $C$. We proceed exactly as in the proof of Lemma~\ref{lem:swapping-sizes}. Let $y$ be any string outside the range of $C^*$. Refer to Figure~\ref{fig:circuit_tree} which gives a diagram of a circuit computing $C^*$; at a layer $i \in [k]$ of this circuit, we have $2^i$ blocks of $n$ bits feeding into $2^i$ copies of $C$, and these copies of $C$ then output $2^{i+1}$ blocks of $n$ bits at the next layer. So working back from the output layer $k$, we can test if any consecutive $2n$-bit block of $y$ is outside of the range of $C$. If none of them are, then we find a preimage for all blocks, interpret this as the output of the previous layer, and continue our search from there. We follow this process all the way back to the input layer or until we find an empty pigeonhole of $C$. If we never find an empty pigeonhole of $C$, then this process will terminate at the input layer with a string $x$ such that $C^*(x) = y$, which is impossible by assumption, so at some point we must indeed find a string outside the range of $C$. Checking whether a particular string is an empty pigeonhole, or finding a preimage if it's not, can be accomplished with one call to a $\mathcal{C}$-inverter oracle by definition. We perform this test at most $2^k = poly(|C|)$ times (once for every copy of $C$ in the diagram in Figure~\ref{fig:circuit_tree}), so overall this process can be accomplished in polynomial time using a $\mathcal{C}$-inverter oracle.
\end{proof}

We now examine the implications of this theorem for particular circuit classes of interest.

\begin{theorem}
For any $0 < \epsilon < \frac{1}{2}$, {\sc $\epsilon$-Hard}$^{\Sigma^{\cl P}_i}$ is complete for {\cl APEPP}$^{\Sigma^{\cl P}_i}$ under $\Delta^{\cl P}_{i+2}$ reductions.
\end{theorem}
\begin{proof}
Containment of {\sc $\epsilon$-Hard}$^{\Sigma^{\cl P}_i}$ in {\cl APEPP}$^{\Sigma^{\cl P}_i}$ follows directly from the proof of Theorem~\ref{thm:truth-tables-in-apepp} with minimal adjustments; we must add the assumption $\epsilon < \frac{1}{2}$ to account for the unbounded fan-in of oracle gates in our counting argument. So it remains only to show that {\sc $\epsilon$-Hard}$^{\Sigma^{\cl P}_i}$ is hard for this class as well. Since the above reduction uses a polynomial number of calls to the $\mathcal{C}$-inverter, it suffices to show that we can implement a $\Sigma^{\cl P}_i$-circuit inverter using a $\Delta^{\cl P}_{i+2}$ oracle. Given this, we can complete the entire reduction in ${\cl P}^{\Delta^{\cl P}_{i+2}} = \Delta^{\cl P}_{i+2}$.

Let $C$ be a $\Sigma^{\cl P}_i$-circuit with $m$ oracle gates, and $y$ be a potential output. To test if $y$ is a valid output, we nondeterministically guess an input $x$, in addition to an output value for every gate in $C$, and a set of witness strings $z_1 \ldots z_m$, one for each of our oracle gates. We then check that each gate output is valid (we will use the guessed witnesses here), and that the value of the terminal gate outputs is $y$. The verification of the terminal gates and all classical $\land$/$\lor$/$\lnot$ gates can be done in polynomial time. Verifying that all $\Sigma^{\cl P}_i$ oracle gates have valid outputs given their inputs corresponds to verifying that a sequence of strings $x_1, \ldots, x_m$ satisfy a sequence of $\Sigma^{\cl P}_i$ and $\Pi^{\cl P}_i$ predicates $\mathcal{P}_1, \ldots, \mathcal{P}_m$, where $m$ is of polynomial length. For each $i$ such that $\mathcal{P}_i$ is a $\Sigma^{\cl P}_i$ predicate, this predicate is of the form $\exists z \mathcal{P}'_i(x_i,z)$ where $\mathcal{P}'_i$ is a $\Pi^{\cl P}_{i-1}$ predicate, and so we can use the $z_i$ we originally guessed and simplify these to $\Pi^{\cl P}_{i-1}$ predicates. For any $i$ such that $\mathcal{P}_i$ is a $\Pi^{\cl P}_i$ predicate we ignore $z_i$. In this way, we can transform all $\mathcal{P}_i$ into $\Pi^{\cl P}_i$ predicates. Verifying that a sequence of of strings satisfies a sequence of $\Pi^{\cl P}_i$ predicates can then be checked with a single $\Pi^{\cl P}_i$ predicate representing their conjunction. So overall the verification process can be carried by checking a single $\Pi^{\cl P}_i$ predicate, and so determining the existence of a solution can be done in $\Sigma^{\cl P}_{i+1}$. From a $\Sigma^{\cl P}_{i+1}$ test to determine the existence of a preimage for $y$, we can compute a preimage when one exists in $\Delta^{\cl P}_{i+2}$ by a standard application of binary search.
\end{proof}

In the absence of any oracle gates, we have the following:
\begin{theorem}\label{thm:apepp-completeness}
For any $0 < \epsilon < 1$, {\sc $\epsilon$-Hard} is complete for {\cl APEPP} under ${\cl P}^{\cl NP}$ reductions.
\end{theorem}

\subsection{Implications of Completeness}
This result gives an \emph{exact algorithmic characterization} of the possibility of proving $2^{\Omega(n)}$ $\Sigma^{\cl P}_i$-circuit lower bounds for ${\cl E}^{\Sigma^{\cl P}_{i+1}}$:
\begin{theorem}\label{thm:LB-algo-equivalence}
There exists a language in ${\cl E}^{\Sigma^{\cl P}_{i+1}}$ with $\Sigma^{\cl P}_i$-circuit complexity $2^{\Omega(n)}$ if and only if there is a $\Delta^{\cl P}_{i+2}$ algorithm for {\sc Empty}$^{\Sigma^{\cl P}_i}$.
\end{theorem}
\begin{proof}
Say there is a language $L$ in ${\cl E}^{\Sigma^{\cl P}_{i+1}}$ with $\Sigma^{\cl P}_i$-circuit complexity $2^{\Omega(n)}$. So there exists an $\epsilon > 0$ such that for all but finitely many $n$, $L$ cannot be computed on length $n$ inputs with $\Sigma^{\cl P}_i$-circuits of size less then $2^{\epsilon n}$. So then we have a polynomial time algorithm for {\sc $\frac{\epsilon}{2}$-Hard}$^{\Sigma^{\cl P}_i}$ as follows: given $1^n$, output the truth table of $L$ over $\lfloor \log n \rfloor$-bit inputs. Since $L \in {\cl E}^{\Sigma^{\cl P}_{i+1}}$, this can be done in $2^{\lfloor \log n \rfloor}2^{O(\log n)} = poly(n)$ time with a $\Sigma^{\cl P}_{i+1}$ oracle. This truth table will have length $\frac{n}{2} \leq 2^{\lfloor \log n \rfloor} \leq n$. We then pad this truth table with 0's at the end to be of length $n$. If there was a circuit of size $n^{\epsilon/2}$ for this $n$-bit truth table on $\lceil \log n \rceil$ bits, then on the first $\lfloor \log n \rfloor$ bits of input this computes the truth table for $L$ on $\lfloor \log n \rfloor$-bit inputs. Since $\frac{n}{2} \leq 2^{\lfloor \log n \rfloor}$, this would imply a circuit of size $2^{\epsilon \lfloor \log n \rfloor}$ to compute $L$ on $\lfloor \log n \rfloor$-bit inputs, contradicting the hardness assumption. Thus, there exists a $\Delta^{\cl P}_{i+2}$ algorithm for {\sc $\frac{\epsilon}{2}$-Hard}$^{\Sigma^{\cl P}_i}$ for some $\epsilon > 0$, and so by Theorem~\ref{thm:apepp-completeness}, there also exists a $\Delta^{\cl P}_{i+2}$ algorithm for {\sc Empty}$^{\Sigma^{\cl P}_i}$.

Alternatively, say there is a $\Delta^{\cl P}_{i+2}$ algorithm for {\sc Empty}$^{\Sigma^{\cl P}_i}$. So in particular there is a $\Delta^{\cl P}_{i+2}$ algorithm for {\sc $\epsilon$-Hard}$^{\Sigma^{\cl P}_i}$ for any fixed $\epsilon < \frac{1}{2}$. Consider the language $L$ decided by the following ${\cl E}^{\Sigma^{\cl P}_{i+1}}$ machine: given an $n$-bit input, we use our $\Delta^{\cl P}_{i+2}$ algorithm for {\sc $\epsilon$-Hard}$^{\Sigma^{\cl P}_i}$ on input $1^{2^n}$ to generate a truth table, then look up the $n$-bit input in this truth table to determine whether to accept or reject. By definition this language must have $\Sigma^{\cl P}_i$-circuit complexity $2^{\Omega(n)}$, and this machine will run in time $poly(2^n) = 2^{O(n)}$ with a $\Sigma^{\cl P}_{i+1}$ oracle.
\end{proof}

In the most interesting case, we conclude that a $2^{\Omega(n)}$ circuit lower bound for ${\cl E}^{\cl NP}$ holds \emph{if and only if} there is a ${\cl P}^{\cl NP}$ algorithm for {\sc Empty}. Together with the results in Section~\ref{sec:constructions-in-apepp}, this gives newfound insight into the difficulty of proving exponential circuit lower bounds for the class ${\cl E}^{\cl NP}$: proving such a lower bound requires solving a \emph{universal explicit construction problem}, and would immediately imply ${\cl P}^{\cl NP}$ constructions for a vast range of combinatorial objects which we currently have no means of constructing without a ${\cl \Sigma_2^P}$ oracle. Theorem~\ref{thm:LB-algo-equivalence} also allows us to derive the following interesting fact about the circuit complexity of ${\cl E}^{\cl NP}$:
\begin{corollary}[Worst-Case to Worst-Case Hardness Amplification in ${\cl E}^{\cl NP}$]\label{cor:e-hard}
If there is a language in ${\cl E}^{\cl NP}$ of circuit complexity $2^{\Omega(n)}$, then there is a language in ${\cl E}^{\cl NP}$ requiring circuits of size $\varbound$.
\end{corollary}
\begin{proof}
By Theorem~\ref{thm:LB-algo-equivalence}, if  there is a language in ${\cl E}^{\cl NP}$ of circuit complexity $2^{\Omega(n)}$, then there is a ${\cl P}^{\cl NP}$ algorithm for {\sc Empty}. By Theorem~\ref{thm:truth-tables-in-apepp}, this implies a ${\cl P}^{\cl NP}$ algorithm for {\sc Hard Truth Table}, and thus a ${\cl P}^{\cl NP}$ construction of a truth table of length $N$ with hardness $\ttbound$. This in turn implies the existence of a language in ${\cl E}^{\cl NP}$ of circuit complexity $\varbound$. 
\end{proof}

Tweaking the proof of Theorem~\ref{thm:epsilon-hard-reduction} slightly we also obtain the following:
\begin{corollary}[Worst-Case to Worst-Case Hardness Amplification in ${\cl EXP}^{\cl NP}$]\label{cor:exp-hard}
If there is a language in ${\cl EXP}^{\cl NP}$ of circuit complexity $2^{n^{\Omega(1)}}$, then there is a language in ${\cl EXP}^{\cl NP}$ requiring circuits of size $\varbound$.
\end{corollary}
\begin{proof}
The proof follows that of the previous corollary, with the following modification to the reduction in Theorem~\ref{thm:epsilon-hard-reduction}: we start with the assumption that for some $\epsilon>0$ we are able to construct $N$-bit truth tables with hardness $2^{\log^{\epsilon}N}$ in time quasipolynomial in $N$ using an {\cl NP} oracle, and then apply the same reduction setting $k = \log^{\lceil \frac{1}{\epsilon} \rceil} |C|$.
\end{proof}

We thus obtain a rather unexpected ``collapse'' theorem for the circuit complexity of ${\cl EXP}^{\cl NP}$: if ${\cl EXP}^{\cl NP}$ has circuits of size $\varbound$ infinitely often, then this class in fact has circuits of size $2^{n^\epsilon}$ infinitely often for every $\epsilon > 0$.

We can refine this slightly as follows.
\begin{definition}
{\sc MCSP}, defined originally in \cite{circuit-minimization}, is the following decision problem: given a truth table $x$ and a size parameter $s$, determine whether $x$ has a circuit of size at most $s$. Let {\sc sMCSP} denote the search variant of this problem, where we are given a truth table $x$ and must output a circuit computing $x$ of minimum size.
\end{definition}
For the hardness amplification procedures in Corollaries \ref{cor:e-hard} and \ref{cor:exp-hard}, we can in fact replace the {\cl NP} oracle with an oracle for {\sc sMCSP}, which is non-trivial since {\sc sMCSP} is not known to be {\cl NP}-hard.
\begin{corollary}\label{cor:smcsp-amp}
If there is a language in ${\cl E}^{\cl sMCSP}$ (resp. ${\cl EXP}^{\cl sMCSP}$)  of circuit complexity $2^{\Omega(n)}$ (resp. $2^{n^{\Omega(1)}}$), then there is a language in ${\cl E}^{\cl sMCSP}$ (resp.  ${\cl EXP}^{\cl sMCSP}$)  requiring circuits of size $\varbound$.
\end{corollary}
\begin{proof}
Recall the two reductions in Lemma~\ref{lem:swapping-sizes} and Theorem~\ref{thm:epsilon-hard-reduction}. In order to find an empty pigeonhole of the input circuit $C$ given a solution to {\sc $\epsilon$-Hard}, we only need to use the $\mathcal{C}$-inverter on $C$ itself. In the case of a reduction from {\sc Hard Truth Table} to {\sc $\epsilon$-Hard}, the circuit of interest $C$ maps circuits of size at most $\ttbound$ to their $N$-bit truth tables, and so an oracle for {\sc sMSCP} would suffice to invert $C$.
\end{proof}
It should be noted that a related result was proven in \cite{circuit-minimization}, showing that this type of hardness amplification is possible in {\cl E} assuming {\sc MCSP}$ \in {\cl P}$. However, their proof does not translate directly to an unconditional result in the oracle setting. Due to their use of the Impagliazzo-Wigderson generator, directly applying their proof in the oracle setting using the relativized generator of \cite{AM-derandomization} would instead show that  if ${\cl E}^{\cl MCSP}$ requires $2^{\Omega(n)}$-sized \emph{nondeterministic} circuits, then ${\cl E}^{\cl MCSP}$ requires $\varbound$-sized \emph{standard} circuits, which is a weaker statement then what is shown above (modulo the search/decision distinction between {\sc sMCSP} and {\sc MCSP}). Another result of a similar flavor was also proven in \cite{Hitchcock-Pavan}, where they establish that, assuming the (unproven) {\cl NP}-completeness of {\sc MCSP}, $2^{n^{\Omega(1)}}$ lower bounds for ${\cl NP} \cap {\cl coNP}$ imply $2^{\Omega(n)}$ lower bounds for ${\cl E}^{\cl NP}$. This type of amplification is incomparable to the amplification demonstrated in Corollaries \ref{cor:e-hard} and \ref{cor:exp-hard}. 

In \cite{hardness-harder}, Buresh-Oppenheim and Santhanam define a notion of ``hardness extraction'' that is highly relevant to the results in this section. Informally, a hardness extractor is a procedure which takes a truth table of length $N$ and circuit complexity $s$, and produces a truth table with nearly maximum circuit complexity relative to its size, whose length is as close to $s$ as possible. The proof of Corollary \ref{cor:smcsp-amp} can in fact be viewed as a construction of a near-optimal hardness extractor using an {\sc sMCSP} oracle. In particular our procedure is able to extract approximately the square root of the input's hardness:
\begin{theorem}\label{thm:hardness-extractor}
There is a polynomial time algorithm using an {\sc sMCSP} oracle which, given a truth table $x$ of length $M$ and circuit complexity $s$, outputs a truth table $y$ of length $N = {\Omega}(\sqrt{\frac{s}{\log M}})$ and circuit complexity $\Omega(\frac{N}{\log N})$.
\end{theorem}
\begin{proof}
The proof follows from a more careful analysis of Corollary~\ref{cor:smcsp-amp}; we give a sketch here. Adapting the proof of Theorem~\ref{thm:truth-tables-in-apepp}, for any $N$ we can efficiently construct a circuit $C$ with $N$ outputs and $\lfloor\frac{N}{2}\rfloor$ inputs, such that any $N$-bit string outside its range requires circuits of size $\delta(\frac{N}{\log N})$ for some fixed $\delta > 0$. 
In particular, it is clear from the proof of Theorem~\ref{thm:truth-tables-in-apepp} that such a $C$ can be constructed of circuit size $O(N^2)$.
Now, let $k$ to be the minimum integer such that $2^k\frac{N}{2} \geq M$. Following the argument in the proof of Theorem~\ref{thm:epsilon-hard-reduction}, we can construct a map $C^*: \{0,1\}^{\lfloor\frac{N}{2}\rfloor} \rightarrow \{0,1\}^{2^k\lfloor\frac{N}{2}\rfloor}$ such that any element of its range has circuit complexity $O(|C|k) = O(N^2k)$, and such that given a string outside the range of $C^*$, we can find a string outside the range of $C$ using $2^k$ calls to an {\sc sMCSP} oracle. Setting $N = \epsilon\sqrt{\frac{s}{\log M}}$ for $\epsilon$ sufficiently small, we get a value of $k \leq \log M$. This in turn means that a circuit of size $O(N^2 k) = O(\epsilon^2 s)$ whose truth table is $x0^{\lfloor\frac{N}{2}\rfloor2^k-M}$ would contradict the fact that $x$ has hardness at least $s$ (for sufficiently small choice of $\epsilon$). Thus, $x0^{\lfloor\frac{N}{2}\rfloor2^k-M}$ must lie outside the range of $C^*$, so using $x$ and our {\sc sMCSP} oracle we can find an $N$-bit string outside the range of $C$. This process requires at most $2^k \leq M$ calls to our {\sc sMCSP} oracle, each of input size at most $N < M$, so overall this takes polynomial time with access to an {\sc sMCSP} oracle.
\end{proof}

A natural goal would be to improve this procedure to extract $(\frac{s}{\log M})^{\frac{1}{2} + \epsilon}$ or ideally $\Omega(\frac{s}{\log M})$ bits of hardness. The only obstacle here is improving the $O(N^2)$ upper bound on the circuit complexity of $C$, the circuit which takes descriptions of $\log N$-input circuits of size $\approx N$ and outputs their truth tables. However, Williams observes in \cite{sat-lb} (see footnote 7 of that paper) that an $N^{2 - \epsilon}$ upper bound on $size(C)$ would imply a (nonuniform) breakthrough for {\sc 3SUM}, so improving this extractor in its current form appears difficult.

\section{Direct {\cl P} Reductions to {\sc Hard Truth Table}}\label{sec:direct-reductions}
Ideally we could extend the completeness result in Theorem~\ref{thm:apepp-completeness} to work with polynomial time reductions, as opposed to ${\cl P}^{\cl NP}$ reductions. However, the {\cl NP} oracle seems highly necessary for the proof techniques used above. Despite this obstacle, we show that there is a natural set of problems in {\cl APEPP} which can be reduced to the problem of finding truth tables of hard functions via {\cl P} reductions. 

A simple way to phrase the following results is that any truth table with sufficiently large circuit complexity will necessarily satisfy a variety of other pseudorandom properties for which no explicit constructions are known, including: rigidity over $\mathbb{F}_2$, high space-bounded communication complexity, and high bit-probe complexity. To show this, we demonstrate that the failure of a string $x$ to possess any of these properties implies a smaller than worst case circuit for $x$. 

To give the tightest reductions possible, we will introduce one new parameterized version of the hard truth table construction problem:

\begin{definition}
{\sc $\delta$-Quite Hard} is the following problem: given $1^N$, output an $N$-bit truth table with hardness $\frac{\delta N}{\log N}$
\end{definition}
This problem is total for sufficiently small $\delta$. We recall also the definition of {\sc $\epsilon$-Hard}, where we must construct a truth table of hardness $N^{\epsilon}$.

\subsection{Rigidity}
We begin with the case of rigidity. We will define the following weaker version of the rigidity construction problem:

\begin{definition}
{\sc $\epsilon$-Rather Rigid} is the following search  problem: given $1^N$, construct an $N \times N$ matrix over $\mathbb{F}_2$ which is $(\epsilon N, \epsilon N^2)$-rigid.
\end{definition}

These parameters are the best possible up to constant factors, and in particular would be sufficient to carry out Valiant's program over $\mathbb{F}_2$.

\begin{theorem}~\label{thm:rigidity-circuit}
For any sufficiently small $\delta > 0$, there exists some $\epsilon > 0$ such that {\sc $\epsilon$-Rather Rigid} reduces in polynomial time to {\sc $\delta$-Quite Hard}.
\end{theorem}

\begin{proof}
To prove this, it suffices to show that for any $N \times N$ matrix $M$ which is not $(\epsilon N, \epsilon N^2)$-rigid, we can construct a boolean circuit with $f(\epsilon)O(\frac{N^2}{\log N})$ gates which decides the value of $M_{i,j}$ given the $2\lceil \log N \rceil$-bit input $(i,j)$, for some function $f$ which approaches zero as $\epsilon$ approaches zero. This then implies that for any fixed $\delta$, an $N^2$-bit truth table requiring circuits of size $\frac{\delta N^2}{\log N}$ must be $(\epsilon N, \epsilon N^2)$-rigid  for some $\epsilon > 0$ which is a function only of $\delta$ (and otherwise determined by $f$ and the constants hidden in the $O(\cdot)$ term).

Say $M$ is not $(\epsilon N, \epsilon N^2)$-rigid. So there exists an $N \times \epsilon N$ matrix $L$, an $\epsilon N \times N$ matrix $R$, and an $N \times N$ matrix $S$ with at most $\epsilon N^2$ nonzero entries, such that $LR \oplus S = M$. We will construct a circuit allowing us to efficiently index $M$ which uses these matrices $L, R, S$ as advice. To encode $L$ and $R$, we can utilize the well-known theorem of Shannon that any truth table of length $N$ can be computed by a circuit of size $O(\frac{N}{\log N})$ \cite{Shannon}. Thus, $L$ can be specified as a list of $\epsilon N$ circuits, each of size $O(\frac{N}{\log N})$, where the $j^{th}$ circuit $L_j$ represents the $j^{th}$ column, and $L_j(i)$ computes $L_{i,j}$. The same can then be done for $R$ (indexing columns instead). To encode $S$, we employ a refinement of Shannon's result due to Lupanov \cite{Coin-Problem}, which tells us that for sufficiently large $N$, any truth table of length $N$ with at most $\epsilon N$ nonzero entries can be computed by circuit of size
\[
\frac{\log\binom{N}{\epsilon N}}{\log \log \binom{N}{\epsilon N}} + o\left(\frac{N}{\log N}\right) \leq H(\epsilon) O(\frac{N}{\log N})
\]
Where $H$ denotes the binary entropy function. Thus $S$ can be computed by a circuit of size $H(\epsilon) O(\frac{N^2}{\log N})$.

It remains to show that the additional circuitry we need to compute $M_{i,j}$ given $i, j$ and the encodings of $L,R,S$ does not increase things too much. By definition, we have that: 
\[M_{i,j} = \langle row_i(L), col_j(R) \rangle \oplus S_{i,j}\]
where the dot product is taken over $\mathbb{F}_2$. Now consider the circuit diagram shown in Figure~\ref{fig:rigid}:

\begin{figure}[H]
    \centering
    \includegraphics[width=0.3\textwidth]{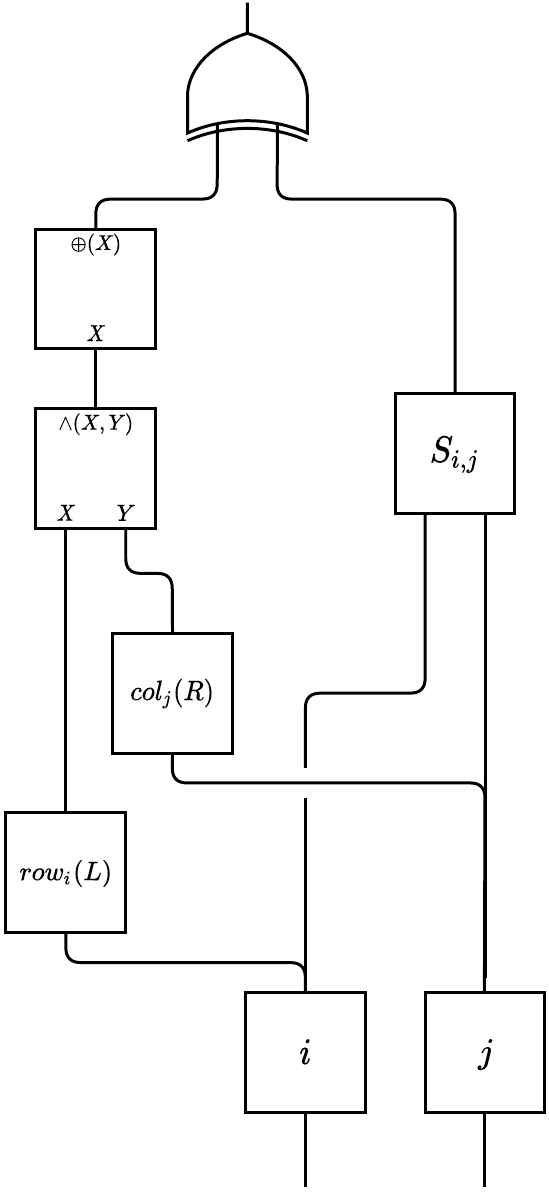}
    \caption{A small circuit for a non-rigid truth table}
    \label{fig:rigid}
\end{figure}
The subcircuit computing $\text{row}_i(L)$ is defined as follows: as described above, $L$ represents a $N \times \epsilon N$ matrix, specified as a list of $\epsilon N$ circuits of size $O(\frac{N}{\log N})$, each computing a column of $L$. $\text{row}_i(L)$ feeds $i$ into each of these circuits in parallel to get an $\epsilon N$-bit string giving value of the $i^{th}$ row of $L$. $\text{col}_j(R)$ is defined analogously. The circuit $S_{i,j}$ simply computes the $(i,j)^{th}$ entry of $S$ given $i,j$. Finally, the subcircuit $\oplus(X)$ computes the parity of its input string, the subcircuit $\land(X,Y)$ computes the bit-wise AND of two equal length strings, and the terminal gate computes the two-bit parity function. Given the previous equation relating the $(i,j)^{th}$ index of $M$ to the $(i,j)^{th}$ rows/columns/indices of $L, R, S$, it is straightforward to see that this circuit performs the necessary calculation.

It is clear that the $\oplus(X)$ and $\land(X,Y)$ can be implemented with a number of gates linear in their input size, which in this case is $\epsilon N$. From the analysis above, each of $\text{row}_i(L)$ and $\text{col}_j(R)$ can be implemented using $\epsilon O(\frac{N^2}{\log N})$ gates, and $S_{i,j}$ can be implemented using $H(\epsilon)  O(\frac{N^2}{\log N})$ gates. So overall, this circuit has size $(H(\epsilon) + \epsilon)O(\frac{N^2}{\log N})$. Since $H(\epsilon) + \epsilon$ approaches zero for decreasing  $\epsilon$, this implies that for any fixed $\delta > 0$, an $N^2$-bit truth table requiring circuits of size $\frac{\delta N^2}{\log N}$ must be $(\epsilon N, \epsilon N^2)$ rigid for some $\epsilon>0$ which is a function only of $\delta$ (and which is otherwise determined by the constants hidden in the $O(\cdot)$ expressions above).
\end{proof}

We thus conclude that if {\cl E} contains a language of circuit complexity $\Omega(\frac{2^n}{n})$, then there is a polynomial time construction of $(\Omega(n), \Omega(n^2))$-rigid matrices over $\mathbb{F}_2$.

\subsection{Space-Bounded Communication Complexity}

The class ${\cl PSPACE^{CC}}$ was defined originally in \cite{communication-classes} as a generalization of the class ${\cl PH^{CC}}$ to an unbounded alternation of quantifiers. We will not give this original definition, but rather a simplification due to \cite{space-communication}.

\begin{definition}
Let $f: \{0,1\}^n \times \{0,1\}^n \rightarrow \{0,1\}$. We say $f$ has a space-$s$ protocol if there is a deterministic protocol deciding $f$ of the following form. Alice receives $x \in \{0,1\}^n$, and Bob receives $y \in \{0,1\}^n$. There is an $s$-bit shared memory, and Alice and Bob alternate turns writing to this memory. On a player's turn, they can modify the contents of the $s$-bit shared memory as a function only of its previous contents and their private input $x$ or $y$, or decide to halt and output some $z \in \{0,1\}$ (this decision is also a function only of their input and the shared memory's previous state). This protocol is valid if $f(xy) = z$ for all $x,y$.
\end{definition}

By a result of Song \cite{space-communication}, we have that if $f  \in {\cl PSPACE^{CC}}$ then $f$ has a $poly(\log n)$ space protocol (Song uses a slightly different model where Alice and Bob have private $s$-bit tapes, but our model is at least as strong up to a doubling in space since sharing the tape only increases their ability to communicate). Due to a basic counting argument, most functions $f$ require $\Omega(n)$ space, and any such function must lie outside of $ {\cl PSPACE^{CC}}$. However, it has been a long standing open problem to give an explicit construction a communication matrix outside of even ${\cl PH^{CC}}$ \cite{space-communication}. We thus define the following search problem:

\begin{definition}
{\sc $\delta$-SPACE} is the following search problem. Given $1^N$, where $N=2^n$, output a communication matrix $\{0,1\}^n \times \{0,1\}^n \rightarrow \{0,1\}$ such that $f$ requires space-$\delta n$ communication protocols.
\end{definition}

\begin{theorem}\label{thm:space-circuit}
For any $\delta < \epsilon < \frac{1}{2}$, {\sc $\delta$-Space} reduces in polynomial time to {\sc $(\frac{1}{2} + \epsilon)$-Hard}.
\end{theorem}

\begin{proof}
The proof follows the exact same strategy as above, constructing a smaller-than-worst-case circuit for any matrix with a $\delta n$-space protocol. We will omit the detailed construction of our circuit since it is very similar to that of the above proof, but give here a clear sketch of its essential structure. Let $f: \{0,1\}^n \times \{0,1\}^n \rightarrow \{0,1\}$, and let $N = 2^n$. We can represent a space $s$ protocol for $f$ as $2s + 4$ different $N \times 2^s$ binary matrices, with each of the first $2s$ matrices telling us how one of the two players will modify a certain cell of the shared memory as a function of its previous state and their input, and the last four matrices determining whether a certain player will halt or continue and the value they will output if they halt, again as a function of their input and the shared memory contents. For $s = \epsilon \log N$ (with $\epsilon < 1$), we can then take these matrices as advice to our circuit, for a total of $\epsilon N^{1 + \epsilon} \log N$ bits of advice. We can then add circuitry to simulate $2^{s+1} = 2N^{\epsilon}$ steps of this protocol on a given input $x,y \in \{0,1\}^n \times \{0,1\}^n$, indexing into the advice matrices in order to determine what to do next using the same indexing constructions we had for the proof of Theorem~\ref{thm:rigidity-circuit}. Any matrix solvable by a space $s$ protocol will be solved by a protocol that halts after $2^{s+1}$ steps (otherwise it will loop forever), so simulating for $2^{s+1}$ steps suffices. Since each of the indexing operations can be implemented using a number of gates linear in $\epsilon N^{1 + \epsilon} \log N$ using the constructions from the proof of Theorem~\ref{thm:rigidity-circuit}, our overall circuit for $f$ will have size $O(2N^{\epsilon}\epsilon N^{1 + \epsilon} \log N) = O(N^{1 + 2\epsilon} \log N)$. Thus, for any $\epsilon < \frac{1}{2}$, a truth table of length $N^2$ requiring circuits of size $N^{1 + 2\epsilon}$ will require $\delta \log N$-space protocols for any $\delta < \epsilon$.
\end{proof}

\subsection{Bit Probe Lower Bounds}
In \cite{Elias-Flower}, Elias and Flower defined a broad model for studying the space/query complexity of data structure problems, known as the ``bit-probe model.'' 
\begin{definition}
Given two sets $D,Q$, a function $f: D \times Q \rightarrow \{0,1\}$ and an integer $b$, the bit-probe complexity of $f$ for space $b$, denoted $BC_b(f)$, is the minimum over all encodings $G: D \rightarrow \{0,1\}^b$ of the number of bits of $G(x)$ that need to be probed in order to determine $f(x,y)$ for the worst case $x \in D,y \in Q$, given access to $y$.
\end{definition}
In this general definition of a data structure problem, we think of $D$ as the set of all possible pieces of ``data'' we might wish to encode in our data structure, $b$ as the number of bits we can use to encode a piece of data, $Q$ as the set of queries we wish to answer about an encoded piece of data, and $f$ as telling us the correct answers to all data/query pairs. $BC_b(f)$ then tells us the minimum number of probes required by any space-$b$ data structure in order to answer every query correctly for every possible piece of data.

This model was investigated further by Miltersen \cite{Bro}, who showed, using a simple counting argument, that most problems require an infeasible amount of space/probes in this model, but pointed out that no explicit data structure problem is known to be infeasible in this sense. We thus define the following explicit construction problem:
\begin{definition}
{\sc $\delta$-Probe} is the following search problem: given $1^N$ where $N = 2^n$ output a truth table $f: \{0,1\}^n \times \{0,1\}^n \rightarrow \{0,1\}$ such that $BC_{2^{\delta n}}(f) > \delta n$.
\end{definition}

\begin{theorem}\label{thm:probe-circuit}
For any $\delta < 2\epsilon < 1$, {\sc $\delta$-Probe} reduces in polynomial time to {\sc $(\frac{1}{2} + \epsilon)$-Hard}.
\end{theorem}
\begin{proof}

Say $f$ is a function $f: D \times Q \rightarrow \{0,1\}$ such that $BC_b(f) \leq k$. So in particular there is an encoding $G: D \rightarrow \{0,1\}^b$ such for any $x \in D$, $y \in Q$, given access to $y$ we can determine $f(x,y)$ using at most $k$ probes to $G(x)$. Let $H: Q \rightarrow [b]^k$ be the function which, given $y$, tells us which positions of $G(x)$ to query. Finally, let $\phi: \{0,1\}^k \times Q \rightarrow \{0,1\}$ be the function which determines $f(x,y)$ given the results of the probes and the value of $y$. We will now give a compact representation for $f$, and then use it to construct a circuit computing $f$.

Let $R$ be a $|D| \times b$ binary matrix whose rows are indexed by elements of $D$, where $R_{i,j}$ gives the value of the $j^{th}$ bit of $G(i)$. Let $S$ be a $|Q| \times k \lceil \log b \rceil$ matrix whose rows are indexed by elements of $Q$, such that the $i^{th}$ row of $S$ is $H(i) \in [b]^k$. Finally, let $Z$ be a $2^k \times |Q|$ matrix such that $Z_{i,j} = \phi(i,j)$.

Now, given $x \in D$, $y \in Q$, we can use $S,R,Z$ to determine $f(x,y)$ as follows. First find the $y^{th}$ row of $S$ to get $H(y)$, which is a list of $k$ indices in $[b]$. Next, find the $x^{th}$ row of $R$, which is the $b$-bit string $G(x)$. Then, probe the indices specified by $H(y)$ to get some $k$-bit string $w$, and finally output $Z_{w, y}$, which is precisely $\phi(w,y) = f(x,y)$.

As in the previous proofs, it can easily be shown that these indexing operations can be accomplished with circuits of size linear in $S,R,Z$. Therefore, $f$ has a circuit of size $O(|S| + |R| + |Z|) = O(|Q|k\log b + |D|b + 2^k|Q|)$. So for any function $D \times Q \rightarrow \{0,1\}^n$ requiring circuits of size $\omega(|Q|k\log b + |D|b + 2^k|Q|)$, we must have $BC_b(f) > k$. 

In particular, if we take $D = Q = \{0,1\}^n$, and for any fixed $\epsilon < 1$ we take $b= 2^{\epsilon n}$, $k = \epsilon n$, we get that for any function $f: \{0,1\}^{2n} \rightarrow \{0,1\}$ requiring circuits of size $\omega(2^{(1 +\epsilon) n})$, it must be the case that $BC_{2^{\epsilon n}}(f) > \epsilon n$, since $|Q|k\log b + |D|b + 2^k|Q| = 2^n \epsilon^2 n^2 + 2^n2^{\epsilon n} + 2^{\epsilon n}2^n = O(2^{(1+\epsilon) n})$. So if we take $\delta < \epsilon$, any truth table of length $N^2$ with hardness $N^{1+\epsilon}$, or in other words any solution to {\sc $(\frac{1}{2} + \frac{\epsilon}{2})$-Hard} on input $1^{N^2}$, must satisfy $BC_{2^{\delta n}}(f) > \delta n$. 
\end{proof}

\subsection{Some Concluding Thoughts} As noted in Section~\ref{sec:constructions-in-apepp}, Theorem~\ref{thm:space-circuit} and Theorem~\ref{thm:probe-circuit} immediately imply that the problems {\sc $\delta$-Space} and {\sc $\delta$-Probe} lie in {\cl APEPP}, since we can construct truth tables with hardness $\ttbound$ in {\cl APEPP} by Theorem~\ref{thm:truth-tables-in-apepp}. Although we phrase these results as reductions, they can also be interpreted as giving conditional polynomial time constructions of various objects, under different circuit lower bound assumptions for the class {\cl E}. In particular, Theorems \ref{thm:space-circuit} and \ref{thm:probe-circuit} establish that if {\cl E} contains a language of circuit complexity $2^{(\frac{1}{2} + \epsilon)n}$ for some $\epsilon > 0$, then polynomial time constructions of hard problems in the bit-probe and space-bounded communication models follow. Theorem~\ref{thm:rigidity-circuit} establishes that if {\cl E} contains a language of circuit complexity $\Omega(\frac{2^n}{n})$, then polynomial time constructions of $(\Omega(n), \Omega(n^2))$-rigid matrices over $\mathbb{F}_2$ follow; such matrices would be sufficient to carry out Valiant's lower bound program. 

It would be interesting as well to find some natural explicit construction problem for which a reduction exists in the opposite direction, i.e. this problem is at least has hard as {\sc Hard Truth Table} or perhaps {\sc $\epsilon$-Hard}. Aside from {\sc $K^{n^2}_U$-Random} for which such a reduction is immediate, we do not know of any other examples. However, we observe the following dichotomy: any explicit construction problem either has a non-trivial algorithm, or is at least as difficult as constructing a somewhat hard truth table. More precisely:
\begin{lemma}
Let $f: \mathbb{N} \rightarrow \mathbb{N}$ be non-increasing, and let $\Pi$ be any property (language) recognizable in complexity class {\cl C}. Let {\sc $\Pi$-Construction} be the search problem: given $1^n$, output an $n$-bit string with property $\Pi$. If {\sc $\Pi$-Construction} is total for sufficiently large $n$, then for sufficiently large $n$ one of the following holds:
\begin{enumerate}
    \item There is a ${\cl TIME(2^{\Tilde{O}(f(n))})}$ algorithm using a {\cl C}-oracle that solves {\sc $\Pi$-Construction} on length $n$ inputs.
    \item There is a polynomial time reduction from the problem of constructing an $n$-bit truth table with circuit complexity $f(n)$ to {\sc $\Pi$-Construction} on length $n$ inputs.
\end{enumerate}
\end{lemma}
\begin{proof}
The proof is a straightforward application of the ``easy witness'' paradigm \cite{easy-witness}. Let $\Pi^n$ be the set of $n$-bit strings with property $\Pi$ and let $H_f^n$ be the set of $n$-bit strings with circuit complexity at most $f(n)$. If $\Pi^n \cap H_f^n = \emptyset$ then any solution to the explicit construction problem for $\Pi$ is necessarily a string with circuit complexity exceeding $f(n)$, and hence a polynomial time reduction from truth table construction to {\sc $\Pi$-Construction} trivially follows. On the other hand, if $\Pi^n \cap H_f^n \neq \emptyset$, we can search over $H_f^n$ for a solution to {\sc $\Pi$-Construction} and will be guaranteed to find a solution. Since $|H_f^n| = 2^{\Tilde{O}(f(n))}$, we can then solve {\sc $\Pi$-Construction} by iterating over $H_f^n$ and using a {\cl C}-oracle to test if each potential solution indeed holds property $\Pi$, in ${\cl TIME(2^{\Tilde{O}(f(n))})}$.
\end{proof}

Setting $f(n) = n^\epsilon$ and combining this with Theorem~\ref{thm:epsilon-hard-reduction}, we conclude:
\begin{theorem}\label{thm:dichotomy}
Let {\sc EC} be any explicit construction problem in {\cl APEPP} (more formally, any search problem in ${\cl STF\Sigma^P_2} \cap {\cl APEPP}$). Then one of the following holds:
\begin{enumerate}
    \item {\sc EC} is {\cl APEPP}-complete under ${\cl P}^{\cl NP}$ reductions.
    \item For every $\epsilon > 0$, {\sc EC} can be solved in time $2^{n^{\epsilon}}$ with an {\cl NP} oracle, for infinitely many $n$.
\end{enumerate}
\end{theorem}

\section{Open Problems}
The most significant question left open in this work is whether {\sc $\epsilon$-Hard} is complete for {\cl APEPP} under polynomial time reductions. One way to demonstrate evidence against this possibility would be to show hardness of {\sc Empty}, perhaps under cryptographic assumptions, since it is widely conjectured that {\sc $\epsilon$-Hard} \emph{does} have a polynomial time algorithm for some $\epsilon > 0$ (this is often cited as the primary reason for believing ${\cl P}={\cl BPP}$ \cite{IW}). More generally, if any sparse search problem is {\cl APEPP} complete under {\cl P} reductions, then {\cl APEPP} lies in {\cl FP/poly}, since we can hard-code all solutions of a fixed polynomial length for the complete sparse problem. It should be noted that the complexity of the dual version of {\sc Empty}, known as {\sc WeakPigeon}, is equivalent to the worst-case complexity of breaking collision-resistant hash functions (in {\sc WeakPigeon} we are given a circuit $C: \{0,1\}^n \rightarrow \{0,1\}^{m}$ with $m < n$, and asked to find a collision \cite{Jerabek-factoring}). A hardness result for {\sc Empty} would be interesting in another respect as well: just as the Natural Proofs barrier \cite{natural-proofs} shows that a generic method of proving circuit lower bounds via a ``Natural Property'' would require solving a hard computational problem, hardness of {\sc Empty} would show that we should not expect to prove exponential lower bounds for {\cl E} via an efficient algorithm that finds an empty pigeonhole for an arbitrary function; something about the \emph{specific} function mapping circuits to their truth tables would have to be utilized.

\section{Acknowledgements}
The author would like to thank Christos Papadimitriou for his guidance and for many inspiring discussions throughout the completion of this work, and Mihalis Yannakakis for his comments on an early draft of this manuscript. The author would also like to thank the anonymous referees for suggesting various improvements to this paper, in particular the addition of Corollary~\ref{cor:exp-hard}, the connection to hardness extractors and the GGM generator, and the simplification of Lemma~\ref{lem:sparse-enc}.
\bibliographystyle{siam}
\bibliography{apepp.bib}

\begin{thebibliography}{10}

\bibitem{alman-chen}
{\sc J.~Alman and L.~Chen}, {\em Efficient construction of rigid matrices using
  an {NP} oracle}, in 2019 IEEE 60th Annual Symposium on Foundations of
  Computer Science (FOCS), 2019, pp.~1034--1055.

\bibitem{communication-classes}
{\sc L.~Babai, P.~Frankl, and J.~Simon}, {\em Complexity classes in
  communication complexity theory}, in 27th Annual Symposium on Foundations of
  Computer Science (sfcs 1986), 1986, pp.~337--347.

\bibitem{BRSW}
{\sc B.~Barak, A.~Rao, R.~Shaltiel, and A.~Wigderson}, {\em 2-source dispersers
  for sub-polynomial entropy and ramsey graphs beating the frankl-wilson
  construction}, in Proceedings of the Thirty-Eighth Annual ACM Symposium on
  Theory of Computing, STOC '06, New York, NY, USA, 2006, Association for
  Computing Machinery, p.~671–680.

\bibitem{pcp-rigidity}
{\sc A.~Bhangale, P.~Harsha, O.~Paradise, and A.~Tal}, {\em Rigid matrices from
  rectangular {PCP}s}, 2020.

\bibitem{hardness-harder}
{\sc J.~Buresh-Oppenheim and R.~Santhanam}, {\em Making hard problems harder},
  21st Annual IEEE Conference on Computational Complexity (CCC'06),  (2006),
  pp.~15 pp.--87.

\bibitem{extractors-survey}
{\sc E.~Chattopadhyay}, {\em Guest column: A recipe for constructing two-source
  extractors}, SIGACT News, 51 (2020), p.~38–57.

\bibitem{flat-sources}
{\sc B.~Chor and O.~Goldreich}, {\em Unbiased bits from sources of weak
  randomness and probabilistic communication complexity}, in 26th Annual
  Symposium on Foundations of Computer Science (sfcs 1985), 1985, pp.~429--442.

\bibitem{data-structure-rigidity}
{\sc Z.~Dvir, A.~Golovnev, and O.~Weinstein}, {\em Static data structure lower
  bounds imply rigidity}, in Proceedings of the 51st Annual ACM SIGACT
  Symposium on Theory of Computing, STOC 2019, New York, NY, USA, 2019,
  Association for Computing Machinery, p.~967–978.

\bibitem{Elias-Flower}
{\sc P.~Elias and R.~A. Flower}, {\em The complexity of some simple retrieval
  problems}, J. ACM, 22 (1975), p.~367–379.

\bibitem{Erdos}
{\sc P.~Erd{\"o}s}, {\em Some remarks on the theory of graphs}, Bulletin of the
  American Mathematical Society, 53 (1947), pp.~292--294.

\bibitem{Friedman}
{\sc J.~Friedman}, {\em A note on matrix rigidity}, Combinatorica, 13 (1993),
  pp.~235--239.

\bibitem{GG}
{\sc E.~Gat and S.~Goldwasser}, {\em Probabilistic search algorithms with
  unique answers and their cryptographic applications}, Electronic Colloquium
  on Computational Complexity (ECCC), 18 (2011), p.~136.

\bibitem{GGM}
{\sc O.~Goldreich, S.~Goldwasser, and S.~Micali}, {\em How to construct random
  functions}, J. ACM, 33 (1986), p.~792–807.

\bibitem{Coin-Problem}
{\sc A.~Golovnev, R.~Ilango, R.~Impagliazzo, V.~Kabanets, A.~Kolokolova, and
  A.~Tal}, {\em Ac0[p] lower bounds against mcsp via the coin problem},
  Electron. Colloquium Comput. Complex., 26 (2019), p.~18.

\bibitem{Hitchcock-Pavan}
{\sc J.~M. Hitchcock and A.~Pavan}, {\em On the np-completeness of the minimum
  circuit size problem}, in FSTTCS, 2015.

\bibitem{Ilango}
{\sc R.~Ilango}, {\em {Approaching MCSP from Above and Below: Hardness for a
  Conditional Variant and $AC^0[p]$}}, in 11th Innovations in Theoretical
  Computer Science Conference (ITCS 2020), T.~Vidick, ed., vol.~151 of Leibniz
  International Proceedings in Informatics (LIPIcs), Dagstuhl, Germany, 2020,
  Schloss Dagstuhl--Leibniz-Zentrum fuer Informatik, pp.~34:1--34:26.

\bibitem{IW}
{\sc R.~Impagliazzo and A.~Wigderson}, {\em {P = BPP} if {E} requires
  exponential circuits: Derandomizing the {XOR} lemma}, in Proceedings of the
  Twenty-Ninth Annual ACM Symposium on Theory of Computing, STOC '97, New York,
  NY, USA, 1997, Association for Computing Machinery, p.~220–229.

\bibitem{Jerabek}
{\sc E.~Jeřábek}, {\em Dual weak pigeonhole principle, boolean complexity,
  and derandomization}, Annals of Pure and Applied Logic, 129 (2004),
  pp.~1--37.

\bibitem{Jerabek-factoring}
\leavevmode\vrule height 2pt depth -1.6pt width 23pt, {\em Integer factoring
  and modular square roots}, Journal of Computer and System Sciences, 82
  (2016), pp.~380--394.

\bibitem{easy-witness}
{\sc V.~Kabanets}, {\em Easiness assumptions and hardness tests: Trading time
  for zero error}, Journal of Computer and System Sciences, 63 (2001),
  pp.~236--252.

\bibitem{circuit-minimization}
{\sc V.~Kabanets and J.-Y. Cai}, {\em Circuit minimization problem}, in
  Proceedings of the Thirty-Second Annual ACM Symposium on Theory of Computing,
  STOC '00, New York, NY, USA, 2000, Association for Computing Machinery,
  p.~73–79.

\bibitem{total}
{\sc R.~Kleinberg, O.~Korten, D.~Mitropolsky, and C.~Papadimitriou}, {\em
  {Total Functions in the Polynomial Hierarchy}}, in 12th Innovations in
  Theoretical Computer Science Conference (ITCS 2021), J.~R. Lee, ed., vol.~185
  of Leibniz International Proceedings in Informatics (LIPIcs), Dagstuhl,
  Germany, 2021, Schloss Dagstuhl--Leibniz-Zentrum f{\"u}r Informatik,
  pp.~44:1--44:18.

\bibitem{AM-derandomization}
{\sc A.~R. Klivans and D.~van Melkebeek}, {\em Graph nonisomorphism has
  subexponential size proofs unless the polynomial-time hierarchy collapses},
  SIAM Journal on Computing, 31 (2002), pp.~1501--1526.

\bibitem{li-extractor}
{\sc X.~Li}, {\em Improved non-malleable extractors, non-malleable codes and
  independent source extractors}, in Proceedings of the 49th Annual ACM SIGACT
  Symposium on Theory of Computing, STOC 2017, New York, NY, USA, 2017,
  Association for Computing Machinery, p.~1144–1156.

\bibitem{AM-SV}
{\sc P.~Miltersen and N.~Vinodchandran}, {\em Derandomizing arthur-merlin games
  using hitting sets}, in 40th Annual Symposium on Foundations of Computer
  Science (Cat. No.99CB37039), 1999, pp.~71--80.

\bibitem{Bro}
{\sc P.~B. Miltersen}, {\em The bit probe complexity measure revisited}, in
  STACS 93, P.~Enjalbert, A.~Finkel, and K.~W. Wagner, eds., Berlin,
  Heidelberg, 1993, Springer Berlin Heidelberg, pp.~662--671.

\bibitem{NW}
{\sc N.~Nisan and A.~Wigderson}, {\em Hardness vs randomness}, Journal of
  Computer and System Sciences, 49 (1994), pp.~149--167.

\bibitem{pseudo-primes}
{\sc I.~C. Oliveira and R.~Santhanam}, {\em Pseudodeterministic constructions
  in subexponential time}, in Proceedings of the 49th Annual ACM SIGACT
  Symposium on Theory of Computing, STOC 2017, New York, NY, USA, 2017,
  Association for Computing Machinery, p.~665–677.

\bibitem{christos}
{\sc C.~H. Papadimitriou}, {\em On the complexity of the parity argument and
  other inefficient proofs of existence}, Journal of Computer and System
  Sciences, 48 (1994), pp.~498 -- 532.

\bibitem{Paris-Wilkie-Woods}
{\sc J.~Paris, A.~Wilkie, and A.~R. Woods}, {\em Provability of the pigeonhole
  principle and the existence of infinitely many primes}, J. Symb. Log., 53
  (1988), pp.~1235--1244.

\bibitem{natural-proofs}
{\sc A.~A. Razborov and S.~Rudich}, {\em Natural proofs}, Journal of Computer
  and System Sciences, 55 (1997), pp.~24--35.

\bibitem{explicit-constructions}
{\sc R.~Santhanam}, {\em The complexity of explicit constructions}, Theory of
  Computing Systems, 51 (2012), pp.~297--312.
\newblock Copyright - Springer Science+Business Media, LLC 2012; Document
  feature - ; Equations; Last updated - 2020-11-18; CODEN - TCSYFI.

\bibitem{MCSP-OWF}
\leavevmode\vrule height 2pt depth -1.6pt width 23pt, {\em Pseudorandomness and
  the minimum circuit size problem}, in 11th Innovations in Theoretical
  Computer Science Conference, {ITCS} 2020, January 12-14, 2020, Seattle,
  Washington, {USA}, T.~Vidick, ed., vol.~151 of LIPIcs, Schloss Dagstuhl -
  Leibniz-Zentrum f{\"{u}}r Informatik, 2020, pp.~68:1--68:26.

\bibitem{Shannon}
{\sc C.~E. {Shannon}}, {\em The synthesis of two-terminal switching circuits},
  The Bell System Technical Journal, 28 (1949), pp.~59--98.

\bibitem{SSS}
{\sc M.~A. Shokrollahi, D.~Spielman, and V.~Stemann}, {\em A remark on matrix
  rigidity}, Information Processing Letters, 64 (1997), pp.~283--285.

\bibitem{space-communication}
{\sc H.~Song}, {\em Space-bounded Communication Complexity}, PhD thesis,
  Tsinghua University, 2014.

\bibitem{Vadhan}
{\sc S.~P. Vadhan}, {\em Pseudorandomness}, vol.~7, Now Delft, 2012.

\bibitem{Valiant}
{\sc L.~G. Valiant}, {\em Graph-theoretic arguments in low-level complexity},
  in Mathematical Foundations of Computer Science 1977, J.~Gruska, ed., Berlin,
  Heidelberg, 1977, Springer Berlin Heidelberg, pp.~162--176.

\bibitem{sat-lb}
{\sc R.~Williams}, {\em Improving exhaustive search implies superpolynomial
  lower bounds}, in Proceedings of the Forty-Second ACM Symposium on Theory of
  Computing, STOC '10, New York, NY, USA, 2010, Association for Computing
  Machinery, p.~231–240.

\bibitem{Yao}
{\sc A.~C. Yao}, {\em Theory and application of trapdoor functions}, in 23rd
  Annual Symposium on Foundations of Computer Science (sfcs 1982), 1982,
  pp.~80--91.

\end{thebibliography}

\end{document}